\newif\iflong
\longtrue
\longfalse

\documentclass[letterpaper,11pt]{article}
\usepackage[margin=1in]{geometry}
\usepackage{amsfonts,amsmath,amssymb,amsthm}
\usepackage{thmtools} 
\usepackage{thm-restate}
\usepackage[noadjust]{cite}
\usepackage{comment,enumerate,hyperref}
\usepackage{cleveref}

\usepackage[colorinlistoftodos,textsize=small,color=red!25!white,obeyFinal]{todonotes}
\newtheorem{theorem}{Theorem}[section]
\newtheorem{corollary}[theorem]{Corollary}

\newtheorem{lemma}[theorem]{Lemma}

\usepackage{soul}
\theoremstyle{definition}
\newtheorem{definition}{Definition}

\theoremstyle{remark}

\usepackage{xcolor,xspace}
\usepackage{booktabs}

\usepackage[boxruled, noend]{algorithm2e}
\usepackage{placeins}

\usepackage{csquotes}

\def\FF  {{\cal F}}
\def\GG  {{\cal G}}

\def\de    {\delta}
\def\al    {\alpha}

\def\la    {\lambda}

\def\empt {\emptyset}
\def\sem  {\setminus}
\def\subs {\subseteq}

\def\SN      {{\sc SND}}                         % Steiner Network
\def\RTF    {{\sc SD-$3$-RSND}}       % Relative $3$-Flow
\def\SEC    {{\sc $4$-Subset $3$-EC}}   % $4$-Subset $3$-Edge-Connectivity

\def\SFC     {{\sc Set Function Cover}}    % Set Function Cover

\title{Improved Approximations for Relative Survivable Network Design}

\author{Michael Dinitz\thanks{Supported in part by NSF awards CCF-1909111 and CCF-2228995.} \\ Johns Hopkins University\\ \texttt{mdinitz@cs.jhu.edu}\and
Ama Koranteng\thanks{Supported in part by an NSF Graduate Research Fellowship and NSF award CCF-1909111}  \\ Johns Hopkins University\\ \texttt{akorant1@jhu.edu} \and Guy Kortsarz\\Rutgers University, Camden\\ \texttt{guyk@camden.rutgers.edu}
\and Zeev Nutov \\The Open  University, Israel  \\\texttt{nutov@openu.ac.il}}

\begin{document} 

\begin{titlepage}
\maketitle
\thispagestyle{empty}

\begin{abstract}
One of the most important and well-studied settings for network design is edge-connectivity requirements. This encompasses uniform demands such as the Minimum $k$-Edge-Connected Spanning Subgraph problem, as well as nonuniform demands such as the Survivable Network Design problem (SND). In a recent paper by [Dinitz, Koranteng, Kortsarz APPROX '22], the authors observed that a weakness of these formulations is that we cannot consider fault-tolerance in graphs that have small cuts but where \emph{some} large fault sets can still be accommodated. To remedy this, they introduced new variants of these problems under the notion \emph{relative} fault-tolerance. Informally, this requires not that two nodes are connected if there are a bounded number of faults (as in the classical setting), but that two nodes are connected if there are a bounded number of faults \emph{and the two nodes are connected in the underlying graph post-faults}.  
%That is, the subgraph we build must ``behave'' identically to the underlying graph with respect to connectivity after bounded faults. The problem is already highly non-trivial even for the case of a single demand.

This seemingly minor change in definition makes the problem dramatically more complex, and so the results in [Dinitz, Koranteng, Kortsarz APPROX '22] are quite limited. For the Relative Survivable Network Design problem (RSND) with non-uniform demands, they are only able to give a nontrivial result when there is a single demand with connectivity requirement $3$---a non-optimal $27/4$-approximation. We strengthen this result in two significant ways: We give a $2$-approximation for RSND where \emph{all requirements are at most $3$}, and a $2^{O(k^2)}$-approximation for RSND with a single demand of \emph{arbitrary value} $k$. To achieve these results, we first use the ``cactus representation'' of minimum cuts to give a lossless reduction to normal SND. Second, we extend the techniques of [Dinitz, Koranteng, Kortsarz APPROX '22] to prove a generalized and more complex version of their structure theorem, which we then use to design a recursive approximation algorithm.
\end{abstract}

\end{titlepage}

\setcounter{page}{1}

%%%%%%%%%%%%%%%%%%%%%
\section{Introduction}
%%%%%%%%%%%%%%%%%%%%%
Fault-tolerance has been a central object of study in approximation algorithms, particularly for network design problems where the graphs we study represent physical objects which might fail (communication links, transportation links, etc.). In these settings it is natural to ask for whatever object we build to be fault-tolerant.  The precise definition of ``fault-tolerance'' varies in different settings, but a common formulation is edge fault-tolerance, which typically takes the form of edge connectivity. Informally, these look like guarantees of the form ``if up to $k$ edges fail, then the nodes I want to be connected are still connected.'' For example, consider the following classical fault-tolerance problem. 

\begin{definition}
In the Survivable Network Design problem (SND, sometimes referred to as Generalized Steiner Network) we are given an edge-weighted graph $G$ and demands $\{(s_i, t_i, k_i)\}_{i \in [\ell]}$, and we are supposed to find the minimum-weight subgraph $H$ of $G$ so that there are at least $k_i$ edge-disjoint paths between $s_i$ and $t_i$ for every $i \in [\ell]$.  In other words, for every $i \in [\ell]$, if fewer than $k_i$ edges fail then $s_i$ and $t_i$ will still be connected in $H$ even after failures.
\end{definition}

The Survivable Network Design problem is well-studied (see~\cite{WGMV95,Jain01,CT00,GGTW09} for a sample); notably, Jain gives a 2-approximation algorithm for the problem in a seminal paper~\cite{Jain01}. Beyond SND, edge fault-tolerance has been studied in many related network design contexts, with the $k$-Edge Connected Spanning Subgraph, Fault-Tolerant Group Steiner Tree, Fault-Tolerant Spanner, and Fault-Tolerant Shortest Paths problems being just a few examples (see~\cite{CT00, KKN12, DR20, BGLP16} for examples). These and other classical fault-tolerance problems, including the Survivable Network Design problem, are \emph{absolute} fault-tolerance problems---if up to \emph{k} objects fail, the remaining graph should function as desired. This differs from the stronger notion of fault-tolerance introduced in~\cite{ap22}, called \emph{relative} fault-tolerance. Relative fault-tolerance makes guarantees that rather than being absolute (``if at most $k$ edges fail the network still functions'') are relative to an underlying graph or system (``if at most $k$ edges fail, the subgraph functions just as well as the original graph post-failures''). 

Relative fault-tolerance is therefore a natural generalization of absolute fault-tolerance: If the input graph has the desired connectivity, then the relative fault-tolerance and absolute fault-tolerance definitions are equivalent.  However, if the input graph does not have the requested connectivity, then relative fault-tolerance allows us to return a solution with interesting and nontrivial guarantees while absolute fault-tolerance forces us to return nothing. In this way, relative fault-tolerance overcomes a significant weakness of absolute fault-tolerance. 

This relative fault-tolerance definition was inspired by a recent line of work on relative notions of fault-tolerance for graph spanners and emulators~\cite{CLPR10,DK11,BDPW18,BP19,DR20,BDR21,BDR22,BDN22,BDN23}. In these settings, the goal is generally to find existential bounds and algorithms to achieve them, rather than to do optimization. In \cite{ap22}, by contrast, their approach takes the point of view of optimization and approximation algorithms. With this notion of fault-tolerance in network design, the authors of \cite{ap22} define the relative version of the Survivable Network Design problem.
%and compare to the instance-specific optimal solution.

\begin{definition} \label{def:RSND}
In the Relative Survivable Network Design problem (RSND), we are given a graph $G = (V, E)$ with edge weights $w : E \rightarrow \mathbb{R}_{\geq 0}$ and demands $\{(s_i, t_i, k_i)\}_{i \in [\ell]}$.  A feasible solution is a subgraph $H$ of $G$ where for all $i \in [\ell]$ and $F \subseteq E$ with $|F| < k_i$, if there is a path in $G \setminus F$ from $s_i$ to $t_i$ then there is also a path in $H \setminus F$ from $s_i$ to $t_i$.  Our goal is to find the minimum weight feasible solution.
\end{definition}

Note that if $s_i$ and $t_i$ are $k_i$-connected in $G$ for every $i \in [\ell]$, then RSND is exactly the same as SND. If in $G$ there exists some $i \in [\ell]$ such that $s_i$ and $t_i$ are \emph{not} $k_i$-connected, then although there is no solution for SND, there is a meaningful RSND solution. 
We note that the fault-tolerance we achieve is really ``one less'' than the given number (there are strict inequalities in the definitions).  This is ``off-by-one'' from the related relative fault-tolerance literature (see the definitions in~\cite{CLPR10}), but makes the connection to the traditional SND formulation cleaner.

There has been recent work on a related network design model introduced by Adjiashvili \cite{ad20flex,ad20ft,boyd22,chekuri22augmenting,bansal22}. In this model, $E$ is partitioned into ``safe'' and ``unsafe'' edges. Informally, in the Flex-SNDP problem we are given a graph $G = (V,E)$ with edge costs and with functions $p,q: V \times V \rightarrow \mathbb{Z}^+$. We must return a min cost subgraph such that for each vertex pair $u,v$, they are $p(u,v)$-connected after deleting any subset of up to $q(u,v)$ unsafe edges. Like RSND, Flex-SNDP is a natural generalization of SND. However, it is an absolute fault-tolerance problem since it does not consider the underlying connectivity of the input. No polynomial-time approximation algorithms are known for general Flex-SNDP, though there has been recent work on several special cases \cite{boyd22,chekuri22augmenting,bansal22, chekuri22racke}.

\subparagraph*{The results of~\cite{ap22}.}
Although relative fault-tolerance is a natural and promising generalization of fault-tolerance, the results given in \cite{ap22} for the RSND problem are quite limited. Outside of a $2$-approximation algorithm for the special case in which all demands are identical, \cite{ap22} is only able to give algorithms for some of the simplest RSND special cases. First, they give an extremely simple $2$-approximation for the RSND special case where all demands are in $\{0,1,2\}$ (also known as the $2$-RSND problem). The algorithm falls out of the observation that there is only a difference between a relative demand of $2$ and an absolute demand of $2$ when there is a cut of size one separating the vertex demand pair. Cuts of size one are very easy to handle, allowing for a simple and straightforward reduction to SND.

Cuts of size two or larger are significantly more difficult to reason about, and so the 2-RSND algorithm does not extend to larger demands. As a result of this more complex cut structure, \cite{ap22} is only able to handle demands of value 3 
(and reason about the size two cuts between them) when there is only \emph{a single demand}, with value 3 (also known as the SD-3-RSND problem). Despite this being an extremely restricted special case of RSND, the algorithm and analysis given by \cite{ap22} are quite complex, depending on a careful graph decomposition involving ``important separators'' (a concept from fixed-parameter tractability~\cite{M11}).  Moreover, this algorithm only achieved a $27/4$-approximation for the problem, far from the $2$-approximation (or even exact algorithm) that one might hope for.

The limited results of \cite{ap22} show that while relative fault-tolerance is an attractive notion, applying it to the Survivable Network Design problem significantly changes the structure of the problem and makes it difficult to reason about and develop algorithms for. For example, while \cite{ap22} only gives a $27/4$-approximation for SD-3-RSND, there is an exact polynomial-time algorithm for the SND equivalent (by a simple reduction to the Min-Cost Flow problem). So one might worry that relative fault-tolerance is simply too difficult of a definition, and the results of~\cite{ap22} are limited precisely because nothing is possible for even slightly more general settings.

\subsection{Our Results and Techniques}
In this paper, we seek to alleviate this worry by providing improved bounds for generalizations of the settings considered in~\cite{ap22}.  In particular, we study two natural generalizations of the SD-3-RSND problem (which~\cite{ap22} provided a $27/4$-approximation for). First, rather than only a single demand with value at most 3, can we handle an \emph{arbitrary} number of demands that are at most 3? Secondly, in the single demand setting, instead of only handling a demand of at most 3, can we generalize to \emph{arbitrary} values?

\subsubsection{3-RSND}
We begin with the setting where all demands are at most $3$, but there can be an arbitrary number of such demands.  We call this the 3-RSND problem.  Note that, as discussed, there are no previous results for this setting, and the most related result is a $27/4$-approximation if there is only a single such demand~\cite{ap22}.  We prove the following theorem.  
\begin{theorem}
\label{t:3SND}
There is a polynomial-time 2-approximation for the 3-RSND problem.
\end{theorem}
% \begin{restatable}{thm}{RSND}
% \label{t:3SND}
% There is a polynomial-time 2-approximation for the 3-RSND problem.
% \end{restatable}

To obtain this theorem, we use entirely different techniques from those used in~\cite{ap22}.  Most notably, we use the \emph{cactus representation} of the global minimum cuts (which in this case are 2-cuts) of the input graph. The cactus representation of global minimum cuts is well studied and has been leveraged in a number of settings (see~\cite{cactus, DN95,P2000,DW98,H1997} for a sample). While it can be defined and constructed for more general connectivity values, for our setting we can construct the cactus representation by contracting components with certain connectivity properties.  This results in a so-called \emph{cactus graph}, which at a high level is a ``tree of cycles'': every pair of cycles intersects on at most one component in the construction.  This cactus graph now has a simple enough structure that it allows us to reduce the original problem to a simpler problem in each of the contracted components.  That is, we are able to show that certain parts of the cactus are essentially ``forced'', while other parts are not necessary, so the only question that remains is what to do ``inside'' of each cactus vertex, i.e., each component.  This reduction makes the connectivity demands inside each component more complicated, but fortunately we are guaranteed 3-connectivity between special vertices inside the component.  Hence we can use Jain's $2$-approximation for SND~\cite{Jain01} without worrying about the relative nature of the demands.  

\subsubsection{SD-\emph{k}-RSND}
Our second improvement is orthogonal: rather than allowing for more demands of at most $3$, we still restrict ourselves to a single demand but allow it to be a general constant $k$ rather than $3$. We call this the SD-$k$-RSND problem. As with the $3$-RSND problem discussed earlier, there are no known results for this problem.
%but for the special case of a single demand of at most $3$ there is a $27/4$-approximation from~\cite{ap22}.  
We prove the following theorem: 
\begin{theorem}
\label{t:SD_RSND}
There is a polynomial-time $2^{O(k^2)}$-approximation for the SD-$k$-RSND problem.
\end{theorem}
% \begin{restatable}{thm}{SD_RSND}
% \label{t:SD_RSND}
% There is a polynomial-time $2^{O(k^2)}$-approximation for the SD-$k$-RSND problem.
% \end{restatable}

To prove this, we extend the technique used by \cite{ap22} for the $k=3$ case. They construct a ``chain'' of 2-separators (cuts of size 2 that are also important separators) so that in each component in the chain, there are no 2-cuts between the incoming separator and the outgoing separator. They are then able to use this structure to characterize the connectivity requirement of any feasible solution restricted to that component. To extend this technique, we use important separators of size up to $k-1$ to carefully construct a \emph{hierarchy} of chains. The hierarchy has $k-1$ levels of nested components, so that for each component in the $i$th level of the hierarchy, there are no cuts of size at most $i$ between the incoming and outgoing separators.  There are multiple ways of constructing such a hierarchy, but we prove that a particular construction yields a hierarchy with a number of useful but delicate properties within a single level and between different levels of the hierarchy. With these properties, we can characterize the complex connectivity requirement of any feasible solution when restricted to a component in the hierarchy. Once we have this structure theorem, we approximate the optimal solution in each component of the hierarchy via a recursive algorithm.

\subsubsection{Simplification of \emph{k}-EFTS}
The $k$-Edge Fault Tolerant Subgraph problem ($k$-EFTS) is the special case of RSND where all demands are identical: every two nodes have a demand of exactly $k$. A $2$-approximation for $k$-EFTS was recently given in~\cite{ap22} via a somewhat complex proof; in particular, they defined a new property called \emph{local weak supermodularity} and used it to show that Jain's iterative rounding still gave the same bounds in the relative setting.  In Appendix~\ref{app:efts} we give a simplification of this proof.  It turns out that local weak supermodularity is not actually needed, and a more classical notion of $\mathcal F$-supermodularity suffices.  This allows us to reduce to previous work in a more black-box manner.  

%%%%%%%%%%%%%%%%%%%%%
\section{Preliminaries}
%%%%%%%%%%%%%%%%%%%%%
We will consider the following special cases of RSND (Definition~\ref{def:RSND}): 
\begin{itemize}
    \item The {\sc $k$-Relative Survivable Network Design} problem ({\sc $k$-RSND}) is the special case of RSND where $r(s,t) \leq k$ for all $s,t \in V$. In this paper we consider the case $k = 3$, namely, the 3-RSND problem.
    \item The {\sc Single Demand $k$-Relative Survivable Network Design} problem ({\sc SD-$k$-RSND}) is the special case of RSND where $r(s,t) = k$ for exactly one vertex pair $s,t \in V$ and there is no demand for any other vertex pairs (equivalently, all other demands are 0). We consider the full SD-$k$-RSND problem for arbitrary $k$.
    \item The {\sc $k$-Edge-Fault-Tolerant-Subgraph} problem ({\sc $k$-EFTS}) is the special case of RSND where $r(s,t) = k$ for all $s,t \in V$.  
\end{itemize}

For each of the listed RSND problem variants, we will use the following notation and definitions throughout.
Let $G=(V,E)$ be a % (connected) 
(multi-)graph and $H$ a spanning subgraph (or an edge subset) of $G$. For $A \subs V$, let $\de_H(A)$ denote the set of edges in $H$ with exactly one endpoint in $A$, and let $d_H(A)=|\de_H(A)|$ be their number. Additionally, let $G[A]$ denote the subgraph of $G$ induced by the vertex set $A$. Let $s,t \in V$. We say that $A$ is an {\bf $st$-set} if $s \in A$ and $t \notin A$,
and that $\de_G(A)$ (or $\de_E(A)$) is an {\bf $st$-cut} of $G$ (induced by $A$). An $st$-cut $\de_G(A)$ (or an $st$-set $A$) is {\bf $G$-minimal} if $\de_G(A)$ contains no other $st$-cut of $G$.
Assuming $G$ is connected, it is easy to see that 
$\de_G(A)$ is $G$-minimal if and only if both $G[A]$ and $G[V \sem A]$ are connected. One can also see that if an $st$-cut $X \subseteq E$ is $G$-minimal, then $X = \de_G(A)$ for some $A \subseteq V$. Finally, let $\lambda_G(s,t)$ denote the size of a min $st$-cut in $G$. 

By Theorem 17 of \cite{ap22}, we may assume without loss of generality that the input graph $G$ of any RSND instance is 2-edge-connected (which we will call ``$2$-connected").

%%%%%%%%%%%%%%%%%%%%%
\section{2-Approximation for 3-RSND (and SD-3-RSND)}
\label{sec:3RSND}
%%%%%%%%%%%%%%%%%%%%%
Given an RSND instance, we say that a vertex pair $s,t$ is a {\bf $k$-demand} if $r(s,t)=k$. We call a $k$-demand {\bf relative} if the minimum $st$-cut has size less than $k$; that is, if $\la_G(s,t) < k$. A $k$-demand is then {\bf ordinary} if $\la_G(s,t) \geq k$. Recall that {\RTF}  has only one demand $st$, and that it is a $3$-demand.   
The edges of any size 2 $st$-cut, or $2$-$st$-cut, belong to any feasible solution so we call them {\bf forced edges}. As a result, we can assume without loss of generality that they have cost $0$. 

\subsection{Overview}
We first give an overview of the theorems and proofs in this section. In order to prove Theorem~\ref{t:3SND}, we will show that we can replace a single {\em relative} $3$-demand by an equivalent set of {\em ordinary} $3$-demands. More formally, we will prove the following.

\begin{theorem} \label{t:D}
Given an {\RTF} instance, % with $2$-connected $G$ 
there exists a polynomially computable set of ordinary $3$-demands, $D$,  
such that for any $H \subs E$ that contains all forced edges, $H$ is a feasible {\RTF} solution if and only if $H$ satisfies all demands in $D$. 
\end{theorem}

This theorem reduces {\RTF} to the ordinary $3$-{\SN} problem (that is, the special case of SND where all demands are at most 3). In fact, this also gives us a lossless reduction from {\sc 3-RSND} to $3$-{\SN}: Given a {\sc 3-RSND} instance, we include the forced edges of all $3$-demands into our solution, replace each relative $3$-demand by an equivalent set of ordinary demands, and obtain an equivalent ordinary $3$-{\SN} instance. Since {\SN} admits approximation ratio $2$, this reduction from {\sc 3-RSND} to $3$-{\SN} implies Theorem~\ref{t:3SND}.

We will also show that {\RTF} is approximation equivalent to certain instances of a special case of $3$-{\SN}. Before we define this special case, we must give a definition. A vertex subset $R$ is a {\bf $k$-edge-connected subset} in a graph $H$ if 
$\la_H(u,v) \geq k$ for all vertex pairs $u,v \in R$. 
Since the relation $\{(u,v) \in V \times V: \mbox{no } (k-1)\mbox{-cut separates } u,v\}$ is transitive,  
this is equivalent to requiring that $\la_H(u,v) \geq k$ for pairs $u,v$ that form a tree on $R$. We will prove that {\RTF} is approximation equivalent to special instances of the following problem: 

\begin{center} \fbox{\begin{minipage}{0.965\textwidth} \noindent
\underline{\SEC}  \\
{\em Input:}  \  \ A graph $J=(V',E')$ with edge costs, and a set $R \subs V'$ of at most $4$ terminals. \\
{\em Output:}   A min-cost subgraph $H$ of $J$, such that $R$ is $3$-edge-connected in $H$. 
\end{minipage}}\end{center}

More specifically, we will prove the following. 

\begin{theorem} \label{t:reduction}
Let $s$ and $t$ be vertices in $J=(V',E')$, where $J$ is the input graph to an instance of {\SEC}. {\RTF} admits approximation ratio $\rho$ if and only if 
{\SEC} %with $2$-edge-connected $H$ and 
with the following properties {\em (A,B)} 
admits approximation ratio $\rho$:
\begin{enumerate}[(A)]
\item
$d_J(s)=d_J(t)= 2$ and $R$ is the set of neighbors of $s,t$.
\item
If $d_{J}(A)=2$ for some $st$-set $A$, then $A=\{s\}$ or $A=V' \sem \{t\}$. Namely, if $F$ is a set of $2$ edges of $J$ such that $J \sem F$ has no $st$-path, then $F=\de_{J}(s)$ or $F=\de_{J}(t)$.  
\end{enumerate}
\end{theorem}

The general {\SEC} problem admits approximation ratio $2$, since it is a special case of {\SN}. However, it is not actually known whether {\SEC} is in P or is NP-hard. To the best of our knowledge, the complexity status of even {\sc $3$-Subset $2$-EC} is open. Both the {\SEC} and {\sc $3$-Subset $2$-EC} problems are related to the special case of {\SN} where the sum of the demands is bounded by a constant. This special case has been studied; for example, both~\cite{CLNV} and~\cite{FML} studied this and related questions. 
Nevertheless, the best known approximation algorithm for the {\SEC} problem gives a 2-approximation. As such, our $2$-approximation for $3$-RSND is the best we can hope for.  In the rest of this section, we prove Theorems \ref{t:reduction}, ~\ref{t:D}, and \ref{t:3SND}. \iflong \else All missing proofs can be found in Appendix~\ref{app:3RSND}. \fi

\subsection{Cactus Representation and Definitions} 
We first give some definitions and describe the cactus representation. The relation $\{(u,v) \in V \times V: \mbox{no } (k-1) \mbox{-cut separates } u,v\}$ is an equivalence, and we will call its equivalence classes {\bf $k$-classes}. We construct a {\bf cactus} $\GG$ by shrinking every nontrivial $3$-class (that is, every $3$-class with at least 2 nodes) of the input graph $G$. Note that since $G$ is $2$-connected, $\GG$ is a connected graph in which every two cycles have at most one node in common; see Fig.~\ref{f:sp}(a) for an example. Going forward, we will identify every $3$-class with the corresponding node of $\GG$.
The edge pairs that belong to the same cycle of $\GG$ are the $2$-cuts of $G$.

\begin{figure}
\centering 
\includegraphics[scale=0.98]{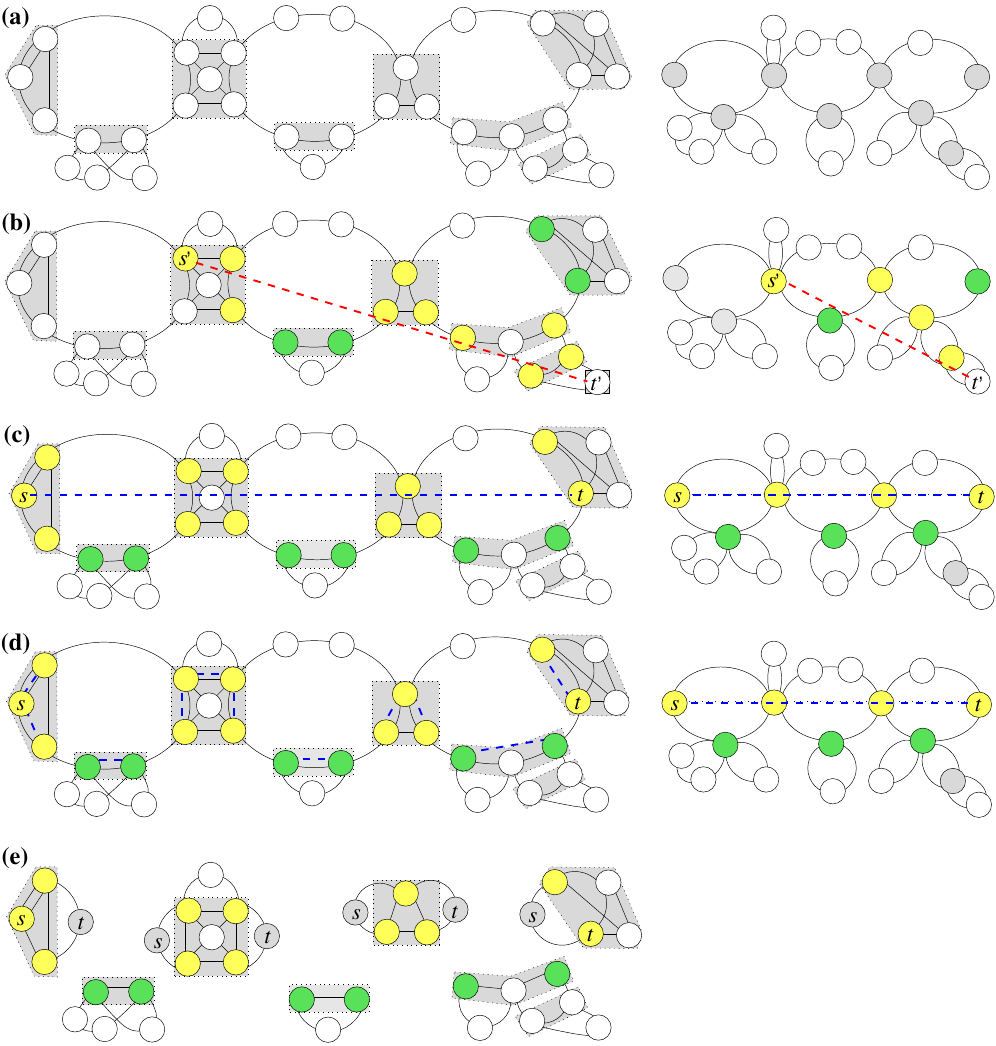}
\caption{
(a)~$3$-classes (shown by gray areas) and the cactus $\GG$ of $G$;
(b) $3$-demands $s't'$ (red dashed line) and the attachment  nodes of $s't'$-relevant $3$-classes, where 
attachment nodes of central classes are yellow and of non-central classes are green;
(c)~$3$-demands $st$ (blue dashed line) and the attachment nodes of $st$-relevant $3$-classes;
(d)~a set of ordinary $3$-demands that is equivalent to the relative $3$-demand $st$.
(e)~components w.r.t. the $3$-demand $st$.}
\label{f:sp}
\end{figure}

We will assume that vertex pair $st$ is a relative $3$-demand (otherwise we have an instance of SD-$3$-{\SN} which can be solved exactly in polynomial time).
%min-cost $3$ edge disjoint $st$-path is an optimal solution.
We say that the {\bf $st$-chain of cycles} of $\GG$ consists of all the cycles (and their nodes) in $\GG$ that contain a $2$-$st$-cut. We say that the nodes, or $3$-classes, on these cycles are {\bf $st$-relevant}. Note that the set of edges in $\GG$ that are in the $st$-chain of cycles are the forced edges. We also say that an $st$-relevant $3$-class is {\bf central} if it contains $s$ or $t$, or belongs to two cycles of the $st$-chain. Additionally, the {\bf attachment nodes} of an $st$-relevant $3$-class are nodes in the $3$-class  that are either $s$ or $t$, or are the ends of the edges (the {\bf attachment edges}) that belong to some cycle in the $st$-chain of cycles. Since $G$ is $2$-connected, the number of attachment nodes in a non-central $3$-class is exactly $2$,
while the number of attachment nodes in a central $3$-class is between $2$ and $4$; see Fig.~\ref{f:sp}(b,c).

\subsection{Proof of Theorems~\ref{t:reduction}, \ref{t:D}, and~\ref{t:3SND}}
For the proof of Theorems~\ref{t:reduction} and \ref{t:D}, we associate with each $st$-relevant $3$-class, $C$, 
a certain graph $G_C$ which we call the {\bf component of $C$}, obtained as follows (see Fig.~\ref{f:sp}(e)):
\begin{itemize}
\item If $C$ is a non-central $3$-class then, in the graph obtained from $G$ by removing the two attachment edges of $C$, $G_C$ is the connected component that contains $C$.

\item If $C$ is a central $3$-class, then removing the attachment edges of $C$ results in at least one and at most two 
connected components that do not contain $C$ -- one contains $s$ and the other contains $t$, if any. 
We obtain $G_C$ from $G$ by contracting the connected component that contains $s$ into node $s$, 
and contracting the connected component that contains $t$ into node $t$. 
\end{itemize}

We now modify the central components $G_C$ to satisfy properties (A,B) from Theorem~\ref{t:reduction}. Consider some central $3$-class $C$, and consider its component $J=G_C$. 
If $J$ does not contain one of the original nodes $s$ or $t$, 
then it has properties (A,B) and no modification is needed.
If $J$ contains the original node $s$, then we rename $s$ to $s'$, add a new node $s$,
and connect new $s$ by two zero cost edges to $s'$. 
The obtained $J$ now has properties (A,B). 
A similar transformation applies if $J$ contains the original node $t$. 

The following lemma is about both the non-central components and these \textit{modified} central components; in the lemma, we show that for $H$ to be a feasible {\RTF} solution, it is necessary and sufficient to satisfy certain connectivity properties within each component. Note that a subgraph $H$ is a feasible solution to SD-$3$-RSND if for any $F \subs E$ with $|F| \leq 2$, the following holds: if there is an $st$-path in $G \setminus F$, then there is an $st$-path in $H \setminus F$.  

\begin{lemma} \label{l:iff}
Let $H $ be a subgraph of $G$, and suppose that $H$ contains all forced edges. Subgraph $H$ is a feasible {\RTF} solution if and only if for every component $J$, the following holds (see Fig,~\ref{f:sp}(e)).
\begin{itemize}
\item[{\em (i)}]
If $J$ is a non-central component, then $H[J]$  contains two edge-disjoint $uv$-paths,
where $u$ and $v$ are the two attachment nodes of $J$. 
\item[{\em (ii)}]
If $J$ is a central component, then $H[J]$ is a feasible solution to the {\RTF} instance in $J$ (with demand $r(s,t)=3$).
\end{itemize}
\end{lemma}
%%%%%%%%iflong%%%%%%%%
\iflong
\begin{proof}
First suppose that $H$ is a feasible {\RTF} solution. We will show that both (i) and (ii) are satisfied.
Suppose to the contrary that (i) does not hold for some non-central component $J$. Let $H' = H[J]$. Let $u$ and $v$ be the attachment nodes of $J$, and assume without loss of generality that in $G$, any (simple) $st$-path has $u$ before $v$, as in Fig,~\ref{f:sp}(a).
Then, there exists a $J$-minimal $uv$-cut, $\de_{H'}(A_J)$, such that $d_{H'}(A_J) \leq 1$ (recall that $H'$ does not include any of the attachment edges of its associated $3$-class). 
Let $A$ be the union of $A_J$ and all components that lie on some (simple) $us$-path of the cycle chain that do not contain $v$. It can be seen that
$d_H(A)=d_{H'}(A_J)+1 \leq 2$. Since $u$ and $v$ are 3-connected in $G$ (they are in the same $3$-class) and $A$ is a $uv$-set, $d_G(A) \geq 3$. It is also not hard to see that $G[A]$ and $G[V \sem A]$ are both connected, and thus $\de_G(A)$ is a $G$-minimal $st$-cut. This all means that $F = \de_{H}(A)$ is a fault set of size at most 2 such that there is an $st$ path in $G \setminus F$, but no $st$ path in $H \setminus F$.
This contradicts the feasibility of $H$.

Now suppose that (ii) does not hold for some central component $J$. Then, there must be a $J$-minimal $st$-cut, $\de_{H'}(A_J)$, such that $d_{H'}(A_J) < \min\{3,d_J(A_J)\}$ (by Lemma~\ref{l:cover}).
Let $A$ be the union of $A_J$ and all components of cycles that precede $J$ in the cactus chain.
Note that $d_H(A)=d_{H'}(A_J)$ and $d_G(A)=d_J(A_J)$. This gives $d_H(A) < \min\{3,d_G(A)\}$, contradicting the feasibility of $H$.

Now suppose that $H$ satisfies (i,ii). We will show that $H$ is a feasible {\RTF} solution.
Let $F \subs E$ such that $G \sem F$ has an $st$-path and $H \sem F$ has no $st$-path.
One can see that since $F$ cannot contain two edges of the same cactus chain cycle, we have that $F$ must have size at least 3; thus, if there is an $st$ path in $G \setminus F$ for $|F| \leq 2$, then there is an $st$ path in $H \setminus F$ as well.
\end{proof}
\fi
%%%%%%%%end iflong%%%%%%%%

Suppose that for the special {\RTF} instances that arise in the central components 
we can achieve approximation ratio $\al$. 
Then, we can achieve ratio $\al$ for general {\RTF} by picking into our solution $H$ three types of edge sets.
\begin{enumerate}
\item
The forced edges.
\item
A min-cost set of $2$ edge-disjoint paths between the attachment nodes of each $st$-relevant non-central component. 
\item
An $\al$-approximate solution in each $st$-relevant central component.
\end{enumerate}

Note that edges picked in steps 1,2 do not invoke any cost in the approximation ratio, 
since by Lemma~\ref{l:iff} we actually pick parts of an optimal solution.
Thus we get that the approximability of {\RTF} is equivalent to the approximability 
of the very special instances that arise in the central components. We will now show that these special instances from the central components are in fact instances of {\SEC} with properties (A,B) from Theorem~\ref{t:reduction}, thus proving Theorem~\ref{t:reduction}.
We will consider only central components with $4$ attachment nodes;
other cases with $3$ or $2$ attachment nodes are similar. 

%${\cal I}=(G=(V,E),c,s,t)$
In what follows, let ${\cal I}$ be an {\RTF} instance on input graph $J$ with properties (A,B) (just as in our central components). 
Let $R=\{x,y,z,w\}$ where $x,y$ are the neighbors of $s$ and $z,w$ are the neighbors of $t$ \iflong(see Fig.~\ref{f:3f})\fi and let $H$ be a subgraph of $J$ that includes the four forced edges $sx,sy,zt$, and $wt$. \iflong We will need the following equivalent characterization of feasible solutions for our restricted case. \else Now we prove that $H$ is a feasible solution for ${\cal I}$ if and only if $R$ is 3-connected in $H$. Lemma~\ref{l:3f} gives the ``only-if'' direction, while Lemma~\ref{l:3f'} gives the ``if'' direction.

%%%%%%%%%%begin iflong%%%%%%%%%%%
\iflong
\begin{lemma} \label{l:char}
Subgraph $H$ is a feasible solution for instance ${\cal I}$ if and only if for any $st$-set $A$ such that $d_H(A)=2$, $\de_H(A)=\de_J(s)$ or $\de_H(A)=\de_J(t)$.
\end{lemma}
\begin{proof}
Let $H$ be a feasible solution for ${\cal I}$ and let $A$ be an $st$-set with $d_H(A)=2$. We will show that either $\de_H(A)=\de_J(s)$ or $\de_H(A)=\de_J(t)$.
Let $F=\de_H(A)$. Then, $H \sem F$ has no $st$-path. Since $H$ is a feasible solution for ${\cal I}$, $J \sem F$ must also have no $st$-path. Component $J$ is $2$-connected, so $\de_H(A)=\de_J(A)$. By property (B) of Theorem~\ref{t:reduction}, we then have that $\de_J(A)=\de_J(s)$ or $\de_J(A)=\de_J(t)$.

Now suppose that for any $st$-set $A$ with $d_H(A)=2$, we have that $\de_H(A)=\de_J(s)$ or  $\de_H(A)=\de_J(t)$. We will show that $H$ is a feasible solution for ${\cal I}$. For any $st$-set $A$ such that $\de_H(A) \neq \de_J(s)$ and $\de_H(A) \neq \de_J(t)$, we have that $d_H(A) \geq 3$. Since $H$ includes the forced edges $sx,sy,zt,$ and $wt$, we also know that $\de_H(s)=\de_J(s)$ and $\de_H(t)=\de_J(t)$. Thus, we have that $d_H(A) \geq \min\{3,d_J(A)\} $ for all $A$ such that $A$ is an $st$-set. By Lemma~\ref{l:cover}, $H$ is a feasible solution for ${\cal I}$.
\end{proof}
\fi
%%%%%%%%%%end iflong%%%%%%%%%%%

\iflong
\begin{figure}
\centering 
\includegraphics{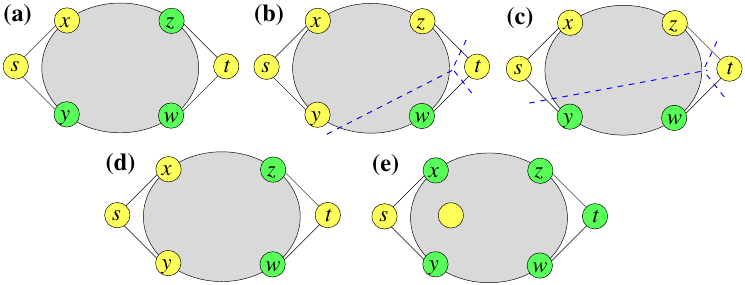}
\caption{Illustration to the proof of Lemmas \ref{l:3f}, \ref{l:3f'}. Nodes in $A$ are yellow and in $V \sem A$ are green.}
\label{f:3f}
\end{figure}
\fi

\iflong Now, we will use Lemma~\ref{l:char} to prove that $H$ is a feasible solution for ${\cal I}$ if and only if $R$ is 3-connected in $H$. Lemma~\ref{l:3f} gives the ``only-if'' direction, while Lemma~\ref{l:3f'} gives the ``if'' direction. \fi

\begin{lemma} \label{l:3f}
If subgraph $H$ is a feasible solution for instance ${\cal I}$, then $R=\{x,y,z,w\}$ is a 3-edge-connected subset in $H$.
\end{lemma}
%%%%%%%begin iflong%%%%%%%
\iflong
\begin{proof}
Suppose that $H$ is a feasible solution for ${\cal I}$, but suppose to the contrary that $R$ is not a $3$-connected subset in $H$. 
Then there is a vertex set $A$ such that $d_H(A) = 2$ and $R \cap A, R \sem A \neq \empt$.
By Lemma~\ref{l:char}, $A$ cannot be an $st$-set.
Hence, without loss of generality, assume that $s,t \in A$. Now we consider several cases, see Fig.~\ref{f:3f}(a,b,c,d).
\begin{enumerate}[(a)]
\item
$|R \cap A| =1$, say $R \cap A=\{x\}$;  see Fig.~\ref{f:3f}(a).
Then $\{sy,tz, tw\} \subs \de_H(A)$, contradicting that $d_H(A) = 2$.
\item
$|R \cap A|=3$, say $R \sem A = \{w\}$; see Fig.~\ref{f:3f}(b). Vertex set $A \sem \{t\}$ is an $st$-set with $d_H(A \sem \{t\})=d_H(A) = 2$, contradicting (by Lemma~\ref{l:char}) that $H$ is a feasible solution.
\item
$|R \cap A|=2$ and $|R \cap \{x,y\}|=|R \cap \{z,w\}|=1$, say $R \cap A=\{x,z\}$; see Fig.~\ref{f:3f}(c).
Then $d_H(A \sem \{t\})=d_H(A) = 2$, contradicting (by Lemma~\ref{l:char}) that $J$ is a feasible solution.
\item
$R \cap A=\{x,y\}$ or $R \cap A=\{z,w\}$, say $R \cap A=\{x,y\}$; see Fig.~\ref{f:3f}(d).
Then $A \sem \{t\}$ is an $st$-set with $d_H(A \sem \{t\}) = d_H(A)-2 = 0$, contradicting that $H$ is a feasible solution.
\end{enumerate}
In all cases we have a contradiction, thus the lemma holds. 
\end{proof}
\fi
%%%%%%%end iflong%%%%%%%

\begin{lemma} \label{l:3f'}
If $R=\{x,y,z,w\}$ is a 3-edge-connected subset in subgraph $H$, then $H$ is a feasible solution for instance ${\cal I}$.
\end{lemma}
%%%%%%%begin iflong%%%%%%%
\iflong
\begin{proof}
Suppose that $R$ is a 3-connected subset in $H$, but suppose to the contrary that $H$ is not a feasible solution.
Then by Lemma~\ref{l:char} there is an $st$-set $A$ with $d_H(A)= 2$ such that $\de_H(A) \neq \de_J(s)$ and $\de_H(A) \neq \de_J(t)$. 
Since $R$ is a $3$-edge-connected subset in $H$, we must have $R \cap A=\empt$ or $R \subset A$. 
Without loss of generality, assume that $R \cap A=\empt$ (see Fig.~\ref{f:3f}(e)); otherwise we interchange the roles of $s$ and $t$. 
But then, $sx,sy \in \de_H(A)$, so we get the contradiction that $\de_H(A)=\{sx,sy\}=\de_J(s)$.
\end{proof}
\fi
%%%%%%%end iflong%%%%%%%

By Lemmas \ref{l:3f} and \ref{l:3f'}, $H$ is a feasible solution for ${\cal I}$ if and only if $H$ includes all forced edges and $R$ is a 3-edge-connected subset---that is, $R$ forms a feasible solution to {\SEC}---in $H$. 
This, along with Lemma~\ref{l:iff}, implies that the approximability of {\RTF} is equivalent 
to that of {\SEC} with properties (A,B), concluding the proof of Theorem~\ref{t:reduction}.

\subparagraph*{Proof of Theorem \ref{t:D}} We will prove that a single relative $3$-demand $st$ can be replaced by an equivalent forest of ordinary $3$-demands in polynomial time, where the trees in this forest span the sets of attachment nodes of the $st$-relevant $3$-classes; see Fig.~\ref{f:sp}(d). 
Recall that by Lemma~\ref{l:iff} and Lemmas \ref{l:3f},\ref{l:3f'}, subgraph $H$ is a feasible {\RTF} solution for 
$3$-demand $st$ if and only if  the following holds for every $st$-relevant $3$-class $C$:

\begin{enumerate}[(i)]
\item
If $C$ is central, then the set $R_C$ of attachment nodes of $C$ is a $3$-connected subset in $H$. 
\item
If $C$ is non-central, then $H[C]$ contains $2$ edge-disjoint $uv$-paths,
where $u$ and $v$ are the attachment nodes of $C$. 
\end{enumerate}
The first condition is equivalent to satisfying a clique of $3$-demands on $R_C$.\footnote{
Recall that since the relation $\{(u,v) \in V \times V: \mbox{no } 2\mbox{-cut separates } u,v\}$ is transitive,
this is equivalent to having a tree of $3$-demands on $R_C$.}
For the second condition, consider a non-central $st$-relevant $3$-class $C$ with attachment nodes $u,v$. 
One can see that if $H$ contains all forced edges and satisfies (i,ii)  
then the number of edge-disjoint $uv$-paths in $H$ is larger by exactly $1$ than their 
number in $H[C]$---the additional path (that exists in $H$ but not in $H[C]$) 
goes along the cycle of the cactus that contains $C$, and there is exactly one such path. Thus, the demand $r(u,v) = 3$ is equivalent to requiring two edge-disjoint paths from $u$ to $v$ in $R_C$ (in addition to including all forced edges).

Thus we obtain an equivalent $3$-{\SN} instance by replacing the single relative $3$-demand $st$ by a set $D$ 
of $3$-demands  that form a clique (or, which is equivalent, a tree) on the set $R_C$ of attachment nodes of every $st$-relevant $3$-class $C$.
These new demands can be computed in polynomial time, and they are ordinary $3$-demands, 
since each $R_C$ is a $3$-edge-connected subset in $G$. 
This concludes the proof of Theorem~\ref{t:D}.

\subparagraph*{Proof of Theorem~\ref{t:3SND}} Finally, we give a $2$-approximation for 3-RSND in the following theorem. We treat each demand in the 3-RSND instance as its own instance of SD-$3$-RSND, solve each SD-$3$-RSND instance, and include the edges of each solution in our output.
%\RSND*

%%%%%%%%%%%%%%%%%%%%%
\section{SD-\emph{k}-RSND} 
%%%%%%%%%%%%%%%%%%%%%
%In the general SD-$k$-RSND problem, we are given a $k$-relative fault tolerance demand for a single vertex pair $(s,t)$, while all other demands are 0. In other words, the set of demands is just $\{(s,t,k)\}$. 
We give a recursive $2^{O(k^2)}$-approximation algorithm for SD-$k$-RSND for arbitrary constant $k$ (Theorem~\ref{t:SD_RSND}). The algorithm is a generalization of the SD-$3$-RSND algorithm from \cite{ap22}. At a high level, the main idea is to partition the input graph using a hierarchy of important separators, prove a structure theorem that characterizes the required connectivity guarantees within each component of the hierarchy, and then achieve these guarantees using a variety of subroutines: a Weighted $st$ Shortest-Path algorithm, the recursive SD-\emph{k}-RSND approximation algorithm, and a Min-Cost Flow algorithm. 

\subsection{Hierarchical Chain Decomposition}
\label{sec:chain_decomp} 
In this section we define important separators and describe how to use them to construct a hierarchical $k$-chain decomposition of $G$. 

\begin{definition}
Let $X$ and $Y$ be vertex sets of a graph $G$. An \textit{$(X,Y)$-separator} of $G$ is a set of edges $S$ such that there is no path between any vertex $x \in X$ and any vertex $y \in Y$ in $G \setminus S$. An $(X,Y)$-separator $S$ is \textit{minimal} if no subset $S' \subset S$ is also an $(X,Y)$-separator.  If $X = \{x\}$ and $Y = \{y\}$, we say that $S$ is an $(x,y)$-separator.
\end{definition}

In \cite{ap22}, the authors provide the following definition.

\begin{definition}[Definition 20 of \cite{ap22}]
Let $S$ be an $(X,Y)$-separator of graph $G$, and let $R$ be the vertices reachable from $X$ in $G \setminus S$. Then $S$ is an \textit{important} $(X,Y)$-separator if $S$ is minimal and there is no $(X,Y)$-separator $S'$ such that $|S'| \leq |S|$ and $R' \subset R$, where $R'$ is the set of vertices reachable from $X$ in $G \setminus S'$.
\end{definition}

In Section 4.1 of \cite{ap22}, the authors describe how to construct the ``$s-t$ 2-chain'' of a graph $G$ \footnote{Note that all separator-based chain definitions given in this section are unrelated to the cactus-based chain definitions in Section \ref{sec:3RSND}.}. Here, we define the \textit{$(X,Y)$ $h$-chain} of $G$ similarly, where $X$ and $Y$ are vertex sets and $h > 0$ is an integer. Instead of using size two important separators to partition the graph, we will use size $h$ important separators. See Figure \ref{fig:chain} for an example.

First, if there are no important $(X,Y)$-separators of size $h$ in $G$, then the $(X,Y)$ $h$-chain of $G$ is simply $G$ and we're done (the chain is a single component, $G$, with no separators). If such an important separator exists, then we first find an important $(X,Y)$-separator $S^h_0$ of size $h$ in $G$, and we let $R^h_0$ be the set of vertices reachable from any vertex $x \in X$ in $G \setminus S^h_0$.  We let $V^h_{(0,r)}$ be the vertices in $R^h_0$ incident on $S^h_0$, and let $V^h_{(1,\ell)}$ be the nodes in $V \setminus R^h_0$ incident on $S^h_0$.  We then proceed inductively.  Given $V^h_{(i,\ell)}$, if there is no important $(V^h_{(i,\ell)}, Y)$-separator of size $h$ in $G \setminus (\cup_{j=0}^{i-1} R^h_j)$ then the chain is finished. Otherwise, let $S^h_i$ be such a separator, let $R^h_i$ be the nodes reachable from $V^h_{(i,\ell)}$ in $(G \setminus (\cup_{j=0}^{i-1} R^h_j)) \setminus S^h_i$, let $V^h_{(i, r)}$ be the nodes in $R^h_i$ incident on $S^h_i$, and let $V^h_{(i+1, \ell)}$ be the nodes in $V \setminus (\cup_{j=0}^i R^h_j)$ incident on $S^h_i$. After this process completes we have our $(X,Y)$ $h$-chain, consisting of components $R^h_0, \dots, R^h_p$ along with important separators $S^h_0, \dots, S^h_{p-1}$ between the components. See Figure \ref{fig:chain}.

\begin{figure}
    \centering
    \includegraphics[scale=0.78]{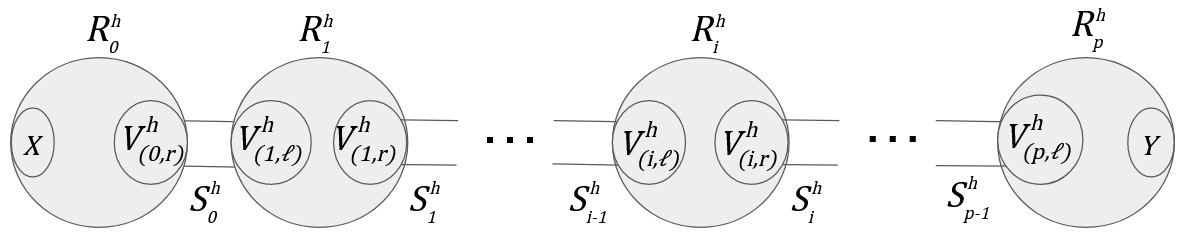}
    \caption{An $(X,Y)$ $h$-chain of a graph $G$, where $h=2$.}
    \label{fig:chain}
\end{figure}

Next we note that by Lemma 17 of \cite{ap22}, we can find an important $(X,Y)$-separator of size $h$ in polynomial time as long as $h$ is a constant.
\begin{lemma}[Lemma 17 of \cite{ap22}] 
\label{lem:imp_sep}
Let $d \geq 0$. An important $(X,Y)$-separator of size $d$ can be found in time $4^d \cdot n^{O(1)}$ (if one exists), where $n = |V|$.
\end{lemma}

\subparagraph*{Constructing the Hierarchical $k$-chain Decomposition.} Now we describe how to construct the \textit{hierarchical $k$-chain decomposition} of $G$. We start by creating the $(s,t)$ 2-chain of $G$. We say that each component of the $(s,t)$ 2-chain is a \textit{2-component} of $G$ in the hierarchical chain decomposition, and that $G$ is the \textit{1-component} of the decomposition. We also say that $G$ is the \textit{parent component} of each 2-component. 
%\iflong Note that the set of 2-components is ordered as follows: The 2-component that contains $s$ is the first component while the 2-component that contains $t$ is last. All other 2-components are adjacent to a left neighbor 2-component and a right neighbor 2-component, via a left important separator and a right important separator, respectively. \fi

We then proceed inductively. Let $R^{h}_i$ be an $h$-component of the hierarchical $k$-chain decomposition. If $h = k-1$, then the decomposition is finished. Otherwise, build the $(V^{h}_{(i, \ell)}, V^{h}_{(i, r)})$ $(h+1)$-chain of $R^{h}_i$. The $(h+1)$-chain consists of $(h+1)$-components. Component $R^{h}_i$ is the parent of these $(h+1)$-components (and these $(h+1)$-components are the children of $R^{h}_i$). Note that the vertex set $V^{h}_{(i, \ell)}$ is also in the first $(h+1)$-component in $R^{h}_i$ and that $V^{h}_{(i, r)}$ is in the last $(h+1)$-component in $R^{h}_i$. After this process completes we have our hierarchical $k$-chain decomposition of $G$. 

The set of all $h$-components can be ordered as follows: The $h$-component that contains $s$ is the first component while the $h$-component that contains $t$ is last. All other $h$-components are adjacent via a left important separator and a right important separator to a left neighbor $h$-component and a right neighbor $h$-component, respectively. Observe that in this hierarchical decomposition, an $h$-component $R^h_i$ has no $(V^{h}_{(i, \ell)}, V^{h}_{(i, r)})$-separator of size $h$ or less.

\subsection{Structure Theorem}
\label{sec:genk_structure}
\subparagraph*{Preliminaries.} We say a subgraph $H$ satisfies the RSND demand $(X,Y,d)$ on input graph $G$ if the following is true: If there is a path from at least one vertex in $X$ to at least one vertex in $Y$ in $G \setminus F$, where $F$ is a set of at most $d-1$ edges, then there is a path from at least one vertex in $X$ to at least one vertex in $Y$ in $H \setminus F$. 
%Recall that for a fixed component $R^h_i$, the edge sets $S^h_{i-1}$ and $S^h_i$ (if they exist) are the edges in an important separator that are incident on vertices in $V^{h}_{(i, \ell)}$ and in $V^{h}_{(i, \ell)}$, respectively. 
Going forward, if $V^{h}_{(i, \ell)} = \{s\}$, then we consider $S^h_{i-1}$ to be the empty set. Similarly, if $V^{h}_{(i, r)} = \{t\}$, then $S^h_i$ is the empty set.

Fix an $h$-component $R^h_i$ and let $X$ be a vertex set such that $X \subseteq V^{h}_{(i, \ell)}$. We say that $S_X$ is the set of edges in $S^h_{i-1}$ incident on vertices in $X$. Similarly, if $Y$ is a vertex set such that $Y \subseteq V^{h}_{(i, r)}$, we say that $S_Y$ is the set of edges in $S^h_i$ incident on vertices in $Y$. We will also use $S$ to denote the set of all edges in an important separator in the hierarchical chain decomposition. Let $H$ be a subgraph of $G$. We will also say that $G^h_i = G[R^h_i]$ is the subgraph of $G$ induced by the vertex set $R^h_i$, and that $H^h_i = H[R^h_i]$ is the subgraph of $H$ induced by $R^h_i$. 

We can now use the hierarchical chain construction to give a structure lemma that characterizes feasible solutions. The lemma states that a subgraph $H$ of $G$ is a feasible solution to SD-$k$-RSND if and only if in the hierarchical $k$-chain decomposition of $G$, all edges in $S$ are in $H$, and certain connectivity requirements between groups of vertices in $V^{h}_{(i, \ell)}$ and in $V^{h}_{(i, r)}$ are met in $H^h_i$ for each component $R^h_i$ in the decomposition. 

\begin{theorem}[Structure Theorem]
\label{thm:genk_characterize}
Subgraph $H$ is a feasible solution to SD-$k$-RSND if and only if all edges in $S$ are included in $H$, and for each $h$-component $R^h_i$ in the hierarchical $k$-chain decomposition of input graph $G$, subgraph $H^h_i$ satisfies the following: 
\begin{enumerate}
    \item $H^h_i$ is a feasible solution to RSND on subgraph $G^h_i$ with demands 
    \begin{gather*}
        \Bigl\{ (X,Y,d) \; : \; X \subseteq V^h_{(i, \ell)}, \; Y \subseteq V^h_{(i, r)}, \; (X,Y) \neq \big(V^h_{(i, \ell)}, V^h_{(i, r)}\big), \\
        \qquad \qquad \qquad d = \max(0,  k  + |S_X| + |S_Y| - |S^h_{i-1}| - |S^h_i|)  \Bigr\}.
    \end{gather*}

    \item $H^h_i$ is a feasible solution to RSND on subgraph $G^h_i$ with demand 
    \begin{gather*}
        \left( V^h_{(i, \ell)}, V^h_{(i, r)}, h+1 \right).
    \end{gather*}

    \item $H^h_i$ is a feasible solution to RSND on subgraph $G^h_i$ with demand
    \begin{gather*}
        \left( V^h_{(i, \ell)}, V^h_{(i, r)}, k-1 \right).
    \end{gather*}   
\end{enumerate}
\end{theorem}

\iflong \else
The proof of this structure theorem is long and involved; it can be found in Appendix~\ref{app:structure}. We describe the proof here at a very high level. To prove the ``only if'' direction, we first assume that we are given some feasible solution $H$.  Then for each of the properties in Theorem~\ref{thm:genk_characterize}, we assume it is false and derive a contradiction by finding a fault set $F \subseteq E$ with $|F| < k$ where there is a path from $s$ to $t$ in $G \setminus F$, but not in $H \setminus F$. The exact construction of such an $F$ depends on which of the properties of Theorem~\ref{thm:genk_characterize} we are analyzing.

To prove the `if'' direction, we assume that $H$ satisfies the conditions of Theorem~\ref{thm:genk_characterize} and consider a fault set $F \subseteq E$ with $|F| < k$ where $s$ and $t$ are connected in $G \setminus F$. We want to show that $s$ and $t$ are connected in $H \setminus F$. We analyze a special ``fault subchain'' of the hierarchical chain decomposition that has specific properties, including that it contains all the edges in $F$. We prove that at least one vertex, $v$, at the end of that subchain is reachable from $s$ in $H \setminus F$ (we also show that $t$ is reachable from $v$ in $H \setminus F$ to complete the proof). 
Specifically, we consider two cases. In Case 1, all fault edges in $F$ are in the same $(k-1)$-level component in the hierarchical chain decomposition; this component is the fault subchain in this case. In Case 2 there is some level, $h$, such that $F$ is not entirely contained in a single $h$-level component, but at level $h-1$, $F$ is completely contained in a single $(h-1)$-level component, $C$. In this case, the fault subchain is the set of all $h$-level components in $C$.

For Case 1, we show that because all components in the hierarchy satisfy Property 2 of Theorem~\ref{thm:genk_characterize}, there must be an $st$ path in $H \setminus F$. The structure of $F$ relative to the hierarchy is more complex in Case 2. As a result we must prove several key characteristics of the components in the fault subchain. With these characteristics, along with Properties 1, 3 of Theorem~\ref{thm:genk_characterize}, we show via a highly technical inductive proof that at least one vertex at the end of the fault subchain is reachable from $s$ in $H \setminus F$, and thus $t$ is also reachable from $s$ in $H \setminus F$.
\fi

\subsection{Algorithm and Analysis} 
\subsubsection{Algorithm}
\label{sec:SD_RSND_algorithm}
We can now use Theorem~\ref{thm:genk_characterize} (Structure Theorem) to give a $2^{O(k^2)}$-approximation algorithm for SD-$k$-RSND. 
Given a graph $G = (V,E)$ with edge weights $w : E \rightarrow \mathbb{R}_{\geq 0}$ and a single demand $\{(s,t,k)\}$, we first create the hierarchical $k$-chain decomposition of $G$ in polynomial time, as described in Section~\ref{sec:chain_decomp}. Within each component we run a set of algorithms to satisfy the RSND demands stated in Theorem~\ref{thm:genk_characterize}. 
%In particular, we will run Weighted $st$ Shortest-Path and Min-Cost Flow algorithms, and we also recursively run the SD-$k$-RSND algorithm with demands in $[2, k-1]$. 
Our solution, $H$, includes the outputs of each of these algorithms along with $S$, the set of all edges in the separators of the hierarchical $k$-chain decomposition. We now describe the set of algorithms run on each component in the hierarchical $k$-chain decomposition. Fix a component $R^h_i$ of the decomposition and let $X \subseteq V^h_{(i,\ell)}$, $Y \subseteq V^h_{(i,r)}$, and $d = \max(0, k + |S_X| + |S_Y| - |S_{i-1}| - |S_{i}|)$:
\begin{itemize}
    \item \textbf{Base Case (Shortest $st$ Path).} For each $X,Y$ pair such that $d = 1$, contract the vertices in $X$ and contract the vertices in $Y$ to create super nodes $x$ and $y$, respectively. We first check in polynomial time if $x$ and $y$ are connected in $G^h_i = G[R^h_i]$. If they are connected, then we create an instance of the Weighted $st$ Shortest-Path problem on $G^h_i$ (in polynomial time), using $x$ and $y$ as our source and destination nodes. For each edge $e \in E(R^h_i)$, set the weight of $e$ to $w(e)$. Run a polynomial-time Weighted $st$ Shortest-Path algorithm on this instance (e.g. Dijkstra's algorithm), and add to $H$ all edges in the output of the algorithm.
    
    \item \textbf{Recursive Step.} For each $X,Y$ pair such that $1 < d < k$, we create an instance of SD-$d$-RSND on $G^h_i$. Contract the vertices in $X$ and contract the vertices in $Y$ to create super nodes $x$ and $y$, respectively. For each edge $e \in E(R^h_i)$, set the cost of $e$ to $w(e)$. The set of RSND demands is just $\{(x,y,d)\}$. Run the recursive polynomial-time SD-$d$-RSND algorithm on this instance, where $d < k$. Add to $H$ all edges in the output of the algorithm.
    
    \item \textbf{Final Recursive Step.} We now create an SD-$(k-1)$-RSND instance on $G^h_i$. Contract the vertices in $V^h_{(i,\ell)}$ and contract the vertices in $V^h_{(i,r)}$ to create super nodes $v_\ell$ and $v_r$, respectively. For each edge $e \in E(R^h_i)$, set the cost of $e$ to $w(e)$. The set of RSND demands is just $\{(v_\ell, v_r, k-1)\}$. Run the recursive SD-$(k-1)$-RSND algorithm on this instance. Add to $H$ all edges in the output of the algorithm.

    \item \textbf{Min-Cost Flow.} Finally, we create an instance of the Min-Cost Flow problem on $G^h_i$. Contract the vertices in $V^h_{(i,\ell)}$ and contract the vertices in $V^h_{(i,r)}$ to create super nodes $v_\ell$ and $v_r$, respectively. Let $v_\ell$ be the source node and $v_r$ be the sink node. For each edge $e \in E(R^h_i)$, set the capacity of $e$ to 1 and set the cost of $e$ to $w(e)$. Require a minimum flow of $h+1$, and run a polynomial-time Min-Cost Flow algorithm on this instance. Since all capacities are integer the algorithm will return an integral flow, so we add to $H$ all edges with non-zero flow.
\end{itemize}

\subsubsection{Analysis} \label{sec:SD_cost}
The following lemma is essentially directly from Theorem~\ref{thm:genk_characterize} (Structure Theorem) and the description of the algorithm. The proof of this lemma can be found in Appendix~\ref{app:cost}.

\begin{lemma}
\label{lem:SD_feasible}
Let $H$ be the output of the algorithm given in Section \ref{sec:SD_RSND_algorithm}. Subgraph $H$ is a feasible solution to the SD-$k$-RSND problem.
\end{lemma}

In the next lemma, we give an approximation ratio for the SD-$k$-RSND algorithm. Let $H^*$ denote the optimal solution, and for any set of edges $A \subseteq E$, let $w(A) = \sum_{e \in A} w(e)$. The next lemma follows from combining the approximation ratios of each of the subroutines used in the algorithm and solving the recurrence.

\begin{lemma}
\label{lem:SD_cost}
$w(H) \leq 2^{O(k^2)}*w(H^*)$.
\end{lemma}

%\iflong
\begin{proof}
Fix level $h$ of the hierarchical $k$-chain decomposition of $G$. Let $H^h_i = H[R^h_i]$, $G^h_i = G[R^h_i]$, and let $H^{h*}_i = H^*[R^h_i]$ be the subgraph of the optimal solution induced by $R^h_i$. Let $d = \max(0,  k  + |S_X| + |S_Y| - |S^h_{i-1}| - |S^h_i|)$.  We also let $W^h_{i,X,Y}$ denote the subgraph of $H^h_i$ returned by the Weighted $st$ Shortest-Path algorithm run on $G^h_i$ for contracted vertex subsets $X \subseteq V^h_{(i, \ell)}$ and $Y \subseteq V^h_{(i, r)}$ such that $d = 1$. Additionally, let $D^h_{i,X,Y}$ denote the subgraph of $H^h_i$ returned by the SD-$d$-RSND algorithm run on $G^h_i$ for contracted $X \subseteq V^h_{(i, \ell)}$ and contracted $Y \subseteq V^h_{(i, r)}$ such that $1<d<k-1$. Let $D^h_{i,\ell,r}$ denote the subgraph of $H^h_i$ returned by the SD-$(k-1)$-RSND algorithm run on $G^h_i$ for contracted $V^h_{(i, \ell)}$ and contracted $V^h_{(i, r)}$ with demand $k-1$. We also let $M^h_i$ denote the subgraph of $H^h_i$ returned by the Min-Cost Flow algorithm run on $G^h_i$ for contracted $V^h_{(i, \ell)}$ and contracted $V^h_{(i, r)}$.

Similarly, let $W^{h*}_{i,X,Y}$ be the optimal solution to the Weighted $st$ Shortest-Path instance on $G^h_i$ for contracted $X \subseteq V^h_{(i, \ell)}$ and contracted $Y \subseteq V^h_{(i, r)}$ such that $d=1$; let $D^{h*}_{i,X,Y}$ be the optimal solution to the SD-$d$-RSND instance on $G^h_i$ for contracted $X \subseteq V^h_{(i, \ell)}$ and contracted $Y \subseteq V^h_{(i, r)}$ such that $1<d<k$; let $D^{h*}_{i,\ell,r}$ be the optimal solution to the SD-$(k-1)$-RSND instance on $G^h_i$ for contracted $V^h_{(i, \ell)}$ and contracted $V^h_{(i, r)}$ with demand $k-1$; let $M^{h*}_i$ be the optimal solution to the Min-Cost Flow instance on $G^h_i$. 

For each $X,Y$ pair such that $d=1$, subgraph $W^h_{i,X,Y}$ is given by an exact algorithm. Subgraph $M^h_i$ is also given by an exact algorithm. Let $T(j)$ be the approximation ratio of our Single Demand RSND algorithm with demand $(s,t,j)$. For each $X,Y$ pair such that $1<d<k$, subgraph $D^h_{i,X,Y}$ is given by our Single Demand RSND algorithm, a $T(d)$-approximation algorithm. Subgraph $D^h_{i,\ell,r}$ is also given by our algorithm, which in this case is a $T(k-1)$-approximation. Hence we have the following for each component $H^h_i$ in the hierarchical $k$-chain decomposition:
\begin{align}
    \label{eq:og1}
    w(W^h_{i,X,Y}) &= w(W^{h*}_{i,X,Y}) && \forall X,Y \text{ s.t. } d=1 \\
    \label{eq:og2}
    w(D^h_{i,X,Y}) &\leq T(d) \cdot w(D^{h*}_{i,X,Y}) && \forall X,Y \text{ s.t. } 1<d<k \\
    w(D^h_{i,\ell,r}) &\leq T(k-1) \cdot w(D^{h*}_{i,\ell,r}) \\
    w(M^h_i) &= w(M^{h*}_i).
\end{align}
Summing over all subsets of $V^h_{(i, \ell)}$ and $V^h_{(i, r)}$, we get the following from expressions (\ref{eq:og1}) and (\ref{eq:og2}):
\begin{align*}
    \sum_{(X,Y): d=1} w(W^h_{i,X,Y}) &= \sum_{(X,Y): d=1} w(W^{h*}_{i,X,Y}) \\
    \sum_{(X,Y): 1<d<k} w(D^h_{i,X,Y}) &\leq \sum_{(X,Y): 1<d<k} T(d) \cdot w(D^{h*}_{i,X,Y}).
\end{align*}
For each level $h$ of the hierarchical $k$-chain decomposition, we denote the set of $h$-components as $R^h_0, R^h_1, \dots R^h_{p_h}$. By summing over all components in the hierarchical $k$-chain decomposition, we get the following:
\begin{align}
\label{eq:sum_comp1}
    \sum_{h=2}^{k-1} \sum_{i=0}^{p_h} \sum_{(X,Y): d=1} w(W^h_{i,X,Y}) &= \sum_{h=2}^{k-1} \sum_{i=0}^{p_h} \sum_{(X,Y): d=1} w(W^{h*}_{i,X,Y}) \\
    \label{eq:sum_comp2}
    \sum_{h=2}^{k-1} \sum_{i=0}^{p_h} \sum_{(X,Y): 1<d<k} w(D^h_{i,X,Y}) &\leq \sum_{h=2}^{k-1} \sum_{i=0}^{p_h} \sum_{(X,Y): 1<d<k}  T(d) \cdot w(D^{h*}_{i,X,Y}) \\
    \label{eq:sum_comp3}
    \sum_{h=2}^{k-1} \sum_{i=0}^{p_h} w(D^h_{i,\ell,r}) &\leq T(k-1) \cdot \sum_{h=2}^{k-1} \sum_{i=0}^{p_h} w(D^{h*}_{i,\ell,r}) \\
    \label{eq:sum_comp4}
    \sum_{h=2}^{k-1} \sum_{i=0}^{p_h} w(M^h_i) &= \sum_{h=2}^{k-1} \sum_{i=0}^{p_h} w(M^{h*}_i).
\end{align}
We also have that
\begin{align}
\label{eq:cost_H}
    w(H^h_i) \leq  \sum_{(X,Y): d=1} w(W^h_{i,X,Y}) + \sum_{(X,Y): 1<d<k} w(D^h_{i,X,Y}) + w(D^h_{i,\ell,r}) + w(M^h_i).
\end{align}
Summing inequality (\ref{eq:cost_H}) over all components in the hierarchical $k$-chain decomposition, then substituting in expressions (\ref{eq:sum_comp1}) through (\ref{eq:sum_comp4}) gives the following:
\begin{align}
    \sum_{h=2}^{k-1} \sum_{i=0}^{p_h} w(H^h_i) &\leq  \sum_{h=2}^{k-1} \sum_{i=0}^{p_h} \sum_{(X,Y): d=1} w(W^h_{i,X,Y}) + \sum_{h=2}^{k-1} \sum_{i=0}^{p_h} \sum_{(X,Y): 1<d<k} w(D^h_{i,X,Y}) \notag\\ &\quad+ \sum_{h=2}^{k-1} \sum_{i=0}^{p_h} w(D^h_{i,\ell,r}) + \sum_{h=2}^{k-1} \sum_{i=0}^{p_h} w(M^h_i) \\ 
    \label{eq:H_ineq}
    &\leq \sum_{h=2}^{k-1} \sum_{i=0}^{p_h} \sum_{(X,Y): d=1} w(W^{h*}_{i,X,Y})  + \sum_{h=2}^{k-1} \sum_{i=0}^{p_h} \sum_{(X,Y): 1<d<k}  T(d) \cdot w(D^{h*}_{i,X,Y}) \notag\\ &\quad+ T(k-1) \cdot \sum_{h=2}^{k-1} \sum_{i=0}^{p_h} w(D^{h*}_{i,\ell,r}) +  \sum_{h=2}^{k-1} \sum_{i=0}^{p_h} w(M^{h*}_i).
\end{align}
The optimal subgraph $H^*$ is a feasible solution, so by Theorem~\ref{thm:genk_characterize}, each RSND demand in the theorem statement must be satisfied in subgraph $H^{h*}_i$, for all $h,i$. For all RSND demands in the theorem statement to be satisfied in a subgraph $H^{h*}_i$, the set of edges $E(H^{h*}_i)$ must be a feasible solution to each of the Weighted $st$ Shortest Path, SD-$k$-RSND, and Min-Cost Flow instances we created on $G^h_i$. Therefore, the cost of $H^{h*}_i$ must be at least the cost of the optimal solution to each of the instances. We have that
\begin{align*}
    w(W^{h*}_{i,X,Y}) &\leq w(H^{h*}_i) && \forall X,Y \text{ s.t. } d=1 \\
    w(D^{h*}_{i,X,Y}) &\leq w(H^{h*}_i) && \forall X,Y \text{ s.t. } 1<d<k \\
    w(D^{h*}_{i,\ell,r}) &\leq w(H^{h*}_i) \\
    w(M^{h*}_i) &\leq w(H^{h*}_i) .
\end{align*}
Let $C(j)$ be the number $X,Y$ pairs such that $X \subseteq V_{(i, \ell)}$, $Y \subseteq V_{(i, r)}$, $(X,Y) \neq (V_{(i, \ell)}, V_{(i, r)})$, and $j = \max(0,  k  + |S_X| + |S_Y| - |S_{i-1}| - |S_i|)$. Summing over subsets of $ V_{(i, \ell)}$ and $V_{(i, r)}$ and over all components in the hierarchical $k$-chain decomposition, we have the following:
\begin{align}
\label{eq:opt1}
    \sum_{h=2}^{k-1} \sum_{i=0}^{p_h} \sum_{(X,Y): d=1} w(W^{h*}_{i,X,Y}) &\leq \sum_{h=2}^{k-1} \sum_{i=0}^{p_h} C(1) \cdot w(H^{h*}_i) \\
    \label{eq:opt2}
    \sum_{h=2}^{k-1} \sum_{i=0}^{p_h} \sum_{(X,Y): 1<d<k} T(d) \cdot w(D^{h*}_{i,X,Y}) &\leq \sum_{h=2}^{k-1} \sum_{i=0}^{p_h} \sum_{(X,Y): 1<d<k} T(d) \cdot w(H^{h*}_i) \notag\\ &\leq \sum_{h=2}^{k-1} \sum_{i=0}^{p_h} \sum_{j=2}^{k-1} C(j) \cdot T(j) \cdot w(H^{h*}_i) \\
    \label{eq:opt3}
    \sum_{h=2}^{k-1} \sum_{i=0}^{p_h}  w(D^{h*}_{i,\ell,r}) &\leq \sum_{h=2}^{k-1} \sum_{i=0}^{p_h} w(H^{h*}_i)\\
    \label{eq:opt4}
    \sum_{h=2}^{k-1} \sum_{i=0}^{p_h} w(M^{h*}_i) &\leq \sum_{h=2}^{k-1} \sum_{i=0}^{p_h} w(H^{h*}_i).
\end{align}
Plugging inequalities (\ref{eq:opt1}) through (\ref{eq:opt4}) into inequality (\ref{eq:H_ineq}) gives the following:
\begin{align*}
    \sum_{h=2}^{k-1} \sum_{i=0}^{p_h} w(H^h_i) &\leq C(1) \sum_{h=2}^{k-1} \sum_{i=0}^{p_h} w(H^{h*}_i) + \sum_{j=2}^{k-1} C(j)T(j) \cdot \sum_{h=2}^{k-1} \sum_{i=0}^{p_h} w(H^{h*}_i) + T(k-1) \sum_{h=2}^{k-1} \sum_{i=0}^{p_h} w(H^{h*}_i) \notag\\ &\quad+ \sum_{h=2}^{k-1} \sum_{i=0}^{p_h} w(H^{h*}_i) \\
    &= \left[ C(1) + \sum_{j=2}^{k-1} C(j)T(j) + T(k-1) + 1 \right] \sum_{h=2}^{k-1} \sum_{i=0}^{p_h} w(H^{h*}_i) 
\end{align*}
Now we account for the important separators, or edges between components, in the hierarchical $k$-chain decomposition. Let $S$ be the separators of the decomposition and let $S^*$ be the set of edges in these separators included in the optimal solution. Edge set $S$ is included in $H$. By Theorem~\ref{thm:genk_characterize}, any feasible solution must include all edges between components in the decomposition. We therefore have that $S = S^*$ and get the following:
\begin{align*}
    w(H) &= \sum_{h=2}^{k-1} \sum_{i=0}^{p_h} w(H^h_i) + w(S) \\ 
    %&\leq \left[ C(1) + \sum_{j=2}^{k-1} C(j)T(j) + T(k-1) + 1 \right] \sum_{h=2}^{k-1} \sum_{i=0}^{p_h} w(H^{h*}_i) + w(S) \\
    &\leq \left[ C(1) + \sum_{j=2}^{k-1} C(j)T(j) + T(k-1) + 1 \right] \sum_{h=2}^{k-1} \sum_{i=0}^{p_h} w(H^{h*}_i) + w(S^*) \\
    &= \left[ C(1) + \sum_{j=2}^{k-1} C(j)T(j) + T(k-1) + 1 \right]w(H^*).
\end{align*}
Therefore, a bound on the approximation ratio of the SD-$k$-RSND algorithm for demand $k$ is given by the following recurrence relation: 
\begin{align*}
    T(k) &\leq C(1) + \sum_{j=2}^{k-1} C(j)T(j) + T(k-1) + 1 \\
    &= \sum_{j=1}^{k-1} C(j)T(j) + T(k-1) + 1 \\
    &\leq \sum_{j=1}^{k-1} C(j) T(k-1) + T(k-1) + 1 
\end{align*}
with base case $T(1) = 1$ (recall that the SD-$k$-RSND problem with $k=1$ is solved exactly using a Weighted $st$ Shortest Paths algorithm). 

First, note that because there are at most $2^{2(k-1)}$ pairs of vertex sets $X,Y$ such that $X \subseteq V^h_{(i, \ell)}$ and $Y \subseteq V^h_{(i, r)}$, we have that $ \sum_{j=1}^{k-1} C(j) \leq 2^{2(k-1)}$. This gives the following:
\begin{align*}
    T(k) &\leq T(k-1)\sum_{j=1}^{k-1} C(j) + T(k-1) + 1 \\
    &\leq T(k-1)\cdot 2^{2(k-1)} + T(k-1) + 1 \\
    &= (2^{2(k-1)}+1)T(k-1) + 1.
\end{align*}
By solving the recurrence relation, we get the following
\begin{align*}
    T(k) &= O\left(\prod^{k-1}_{x=1} 2^{2(k-x)}\right) \\
    &= O\left(2^{\sum_{j=1}^k 2j}\right) \\
    &= O\left(2^{k(k+1)}\right) \\
    &= 2^{O(k^2)}.
\end{align*}

Therefore, we have that $w(H) \leq 2^{O(k^2)} \cdot w(H^*)$.
\end{proof}
%\fi

Theorem~\ref{t:SD_RSND} is directly implied by Lemmas~\ref{lem:SD_feasible} and \ref{lem:SD_cost} together with the observation that the algorithm runs in polynomial time.

\newpage
\bibliography{refs}

\begin{thebibliography}{10}

\bibitem{ad20flex}
David Adjiashvili, Felix Hommelsheim, and Moritz Mühlenthaler.
\newblock Flexible graph connectivity: Approximating network design problems
  between 1- and 2-connectivity, 2020.
\newblock \href {http://arxiv.org/abs/1910.13297} {\path{arXiv:1910.13297}}.

\bibitem{ad20ft}
David Adjiashvili, Felix Hommelsheim, Moritz Mühlenthaler, and Oliver Schaudt.
\newblock Fault-tolerant edge-disjoint paths -- beyond uniform faults, 2020.
\newblock \href {http://arxiv.org/abs/2009.05382} {\path{arXiv:2009.05382}}.

\bibitem{bansal22}
Ishan Bansal, Joseph Cheriyan, Logan Grout, and Sharat Ibrahimpur.
\newblock Improved approximation algorithms by generalizing the primal-dual
  method beyond uncrossable functions, 2022.
\newblock \href {http://arxiv.org/abs/2209.11209} {\path{arXiv:2209.11209}}.

\bibitem{BGLP16}
Davide Bilò, Luciano Gualà, Stefano Leucci, and Guido Proietti.
\newblock Multiple-edge-fault-tolerant approximate shortest-path trees, 2016.
\newblock URL: \url{https://arxiv.org/abs/1601.04169}, \href
  {https://doi.org/10.48550/ARXIV.1601.04169}
  {\path{doi:10.48550/ARXIV.1601.04169}}.

\bibitem{BDN22}
Greg Bodwin, Michael Dinitz, and Yasamin Nazari.
\newblock Vertex fault-tolerant emulators.
\newblock In Mark Braverman, editor, {\em 13th Innovations in Theoretical
  Computer Science Conference, {ITCS} 2022}, volume 215 of {\em LIPIcs}, pages
  25:1--25:22. Schloss Dagstuhl - Leibniz-Zentrum f{\"{u}}r Informatik, 2022.
\newblock \href {https://doi.org/10.4230/LIPIcs.ITCS.2022.25}
  {\path{doi:10.4230/LIPIcs.ITCS.2022.25}}.

\bibitem{BDN23}
Greg Bodwin, Michael Dinitz, and Yasamin Nazari.
\newblock {Epic Fail: Emulators can tolerate polynomially many edge faults for
  free}.
\newblock In {\em 14th Innovations in Theoretical Computer Science Conference,
  {ITCS} 2023}, 2023.

\bibitem{BDPW18}
Greg Bodwin, Michael Dinitz, Merav Parter, and Virginia~Vassilevska Williams.
\newblock Optimal vertex fault tolerant spanners (for fixed stretch).
\newblock In Artur Czumaj, editor, {\em Proceedings of the Twenty-Ninth Annual
  {ACM-SIAM} Symposium on Discrete Algorithms, {SODA} 2018, New Orleans, LA,
  USA, January 7-10, 2018}, pages 1884--1900. {SIAM}, 2018.

\bibitem{BDR21}
Greg Bodwin, Michael Dinitz, and Caleb Robelle.
\newblock Optimal vertex fault-tolerant spanners in polynomial time.
\newblock In {\em Proceedings of the Thirty-Second Annual {ACM-SIAM} Symposium
  on Discrete Algorithms, {SODA} 2021}, 2021.

\bibitem{BDR22}
Greg Bodwin, Michael Dinitz, and Caleb Robelle.
\newblock Optimal vertex fault-tolerant spanners in polynomial time.
\newblock In Joseph~(Seffi) Naor and Niv Buchbinder, editors, {\em Proceedings
  of the 2022 {ACM-SIAM} Symposium on Discrete Algorithms, {SODA} 2022}, pages
  2924--2938. {SIAM}, 2022.
\newblock \href {https://doi.org/10.1137/1.9781611976465.174}
  {\path{doi:10.1137/1.9781611976465.174}}.

\bibitem{BP19}
Greg Bodwin and Shyamal Patel.
\newblock A trivial yet optimal solution to vertex fault tolerant spanners.
\newblock In {\em Proceedings of the 2019 ACM Symposium on Principles of
  Distributed Computing}, PODC ’19, page 541–543, New York, NY, USA, 2019.
  Association for Computing Machinery.
\newblock \href {https://doi.org/10.1145/3293611.3331588}
  {\path{doi:10.1145/3293611.3331588}}.

\bibitem{boyd22}
Sylvia Boyd, Joseph Cheriyan, Arash Haddadan, and Sharat Ibrahimpur.
\newblock Approximation algorithms for flexible graph connectivity, 2022.
\newblock \href {http://arxiv.org/abs/2202.13298} {\path{arXiv:2202.13298}}.

\bibitem{CLPR10}
Shiri Chechik, Michael Langberg, David Peleg, and Liam Roditty.
\newblock Fault tolerant spanners for general graphs.
\newblock {\em {SIAM} J. Comput.}, 39(7):3403--3423, 2010.

\bibitem{chekuri22racke}
Chandra Chekuri and Rhea Jain.
\newblock Approximating flexible graph connectivity via r\"acke tree based
  rounding, 2022.
\newblock \href {http://arxiv.org/abs/2211.08324} {\path{arXiv:2211.08324}}.

\bibitem{chekuri22augmenting}
Chandra Chekuri and Rhea Jain.
\newblock Augmentation based approximation algorithms for flexible network
  design, 2022.
\newblock \href {http://arxiv.org/abs/2209.12273} {\path{arXiv:2209.12273}}.

\bibitem{CLNV}
J.~Cheriyan, B.~Laekhanukit, G.~Naves, and A.~Vetta.
\newblock Approximating rooted steiner networks.
\newblock {\em ACM Trans. Algorithms}, 11(2):8:1--8:22, 2014.

\bibitem{CT00}
Joseph Cheriyan and Ramakrishna Thurimella.
\newblock Approximating minimum-size k-connected spanning subgraphs via
  matching.
\newblock {\em SIAM Journal on Computing}, 30(2):528--560, 2000.
\newblock \href {https://doi.org/10.1137/S009753979833920X}
  {\path{doi:10.1137/S009753979833920X}}.

\bibitem{ap22}
M.~Dinitz, A.~Koranteng, and G.~Kortsarz.
\newblock Relative survivable network design.
\newblock In {\em APPROX-RANDOM}, volume 245, pages 41:1--41:19, 2022.

\bibitem{DK11}
Michael Dinitz and Robert Krauthgamer.
\newblock Fault-tolerant spanners: better and simpler.
\newblock In {\em Proceedings of the 30th Annual {ACM} Symposium on Principles
  of Distributed Computing, {PODC} 2011, San Jose, CA, USA, June 6-8, 2011},
  pages 169--178, 2011.

\bibitem{DR20}
Michael Dinitz and Caleb Robelle.
\newblock Efficient and simple algorithms for fault-tolerant spanners.
\newblock In Yuval Emek and Christian Cachin, editors, {\em {PODC} '20: {ACM}
  Symposium on Principles of Distributed Computing}, pages 493--500. {ACM},
  2020.
\newblock \href {https://doi.org/10.1145/3382734.3405735}
  {\path{doi:10.1145/3382734.3405735}}.

\bibitem{DW98}
Ye~Dinitz and Jeffery Westbrook.
\newblock Maintaining the classes of 4-edge-connectivity in a graph on-line.
\newblock {\em Algorithmica}, 20:242--276, 1998.

\bibitem{DN95}
Yefim Dinitz and Zeev Nutov.
\newblock A 2-level cactus model for the system of minimum and minimum+ 1
  edge-cuts in a graph and its incremental maintenance.
\newblock In {\em Proceedings of the twenty-seventh annual ACM symposium on
  Theory of computing}, pages 509--518, 1995.

\bibitem{FML}
A.~E. Feldmann, A.~Mukherjee, and E.~J. van Leeuwen.
\newblock The parameterized complexity of the survivable network design
  problem.
\newblock In {\em SOSA}, pages 37--56, 2022.

\bibitem{GG12}
Harold~N. Gabow and Suzanne~R. Gallagher.
\newblock Iterated rounding algorithms for the smallest k-edge connected
  spanning subgraph.
\newblock {\em SIAM Journal on Computing}, 41(1):61--103, 2012.
\newblock \href {https://doi.org/10.1137/080732572}
  {\path{doi:10.1137/080732572}}.

\bibitem{GGTW09}
Harold~N. Gabow, Michel~X. Goemans, {\'{E}}va Tardos, and David~P. Williamson.
\newblock Approximating the smallest \emph{k}-edge connected spanning subgraph
  by lp-rounding.
\newblock {\em Networks}, 53(4):345--357, 2009.

\bibitem{H1997}
Monika~Rauch Henzinger.
\newblock A static 2-approximation algorithm for vertex connectivity and
  incremental approximation algorithms for edge and vertex connectivity.
\newblock {\em Journal of Algorithms}, 24(1):194--220, 1997.

\bibitem{Jain01}
Kamal Jain.
\newblock A factor 2 approximation algorithm for the generalized steiner
  network problem.
\newblock {\em Combinatorica}, 21(1):39--60, 2001.
\newblock \href {https://doi.org/10.1007/s004930170004}
  {\path{doi:10.1007/s004930170004}}.

\bibitem{KKN12}
Rohit Khandekar, Guy Kortsarz, and Zeev Nutov.
\newblock Approximating fault-tolerant group-steiner problems.
\newblock {\em Theoretical Computer Science}, 416:55--64, 2012.

\bibitem{cactus}
On{-}Hei~Solomon Lo, Jens~M. Schmidt, and Mikkel Thorup.
\newblock Compact cactus representations of all non-trivial min-cuts.
\newblock {\em Discret. Appl. Math.}, 303:296--304, 2021.

\bibitem{M11}
D{\'a}niel Marx.
\newblock Important separators and parameterized algorithms.
\newblock In {\em International Workshop on Graph-Theoretic Concepts in
  Computer Science}, pages 5--10. Springer, 2011.

\bibitem{P2000}
Johannes A.~La Poutre.
\newblock Maintenance of 2- and 3-edge-connected components of graphs ii.
\newblock {\em SIAM J. Comput.}, 29(5):1521--1549, 2000.

\bibitem{WGMV95}
David~P. Williamson, Michel~X. Goemans, Milena Mihail, and Vijay~V. Vazirani.
\newblock A primal-dual approximation algorithm for generalized steiner network
  problems.
\newblock {\em Combinatorica}, 15(3):435--454, 1995.
\newblock \href {https://doi.org/10.1007/BF01299747}
  {\path{doi:10.1007/BF01299747}}.

\end{thebibliography}

\appendix

\iflong \else 
\section{Proofs from Section \ref{sec:3RSND} (3-RSND)} \label{app:3RSND}
First we must show that any RSND demand can be expressed as a set of cut-covering constraints.

\begin{lemma}
\label{l:cover}
An RSND demand $r(s,t)=k$ is satisfied by a subgraph $H$ if and only if for any $G$-minimal $st$-set $A$, $d_{H}(A) \geq \min\{k,d_G(A)\}$
\end{lemma}
\begin{proof}
First we prove the ``only if'' direction. Suppose that demand $r(s,t)=k$ is satisfied by subgraph $H$, but suppose for the sake of contradiction that there is some $G$-minimal $st$-set, $A$, such that $d_{H}(A) < \min\{k,d_G(A)\}$. Then, $F=\de_{H}(A)$ is a fault set of size at most $k-1$ such that there is an $st$-path in $G \setminus F$ (since $G[A]$ and $G[V\setminus A]$ are connected and $|\de_G(A) \setminus F| \geq 1$), but no $st$-path in $H \setminus F$. This contradicts the feasibility of $H$.

Now we show the ``if'' direction. Suppose that for any $G$-minimal $st$-set $A$, 
$d_{H}(A) \geq \min\{k,d_G(A)\}$. 
Suppose however for the sake of contradiction that $H$ does not satisfy requirement $r(s,t) = k$. 
Then there is an $H$-minimal $st$-cut $X \subseteq E$ with $|X| < k$ such that $s,t$ are not connected in $H \setminus X$ but are connected in $G \setminus X$. Edge set $X$ is $H$-minimal, so we have that $X = \de_H(A)$ for some vertex set $A$. Note that $\de_G(A)$ is a $G$-minimal $st$-set with $d_G(A) > d_H(A)$. By our assumption, we have that $d_{H}(A) \geq \min\{k,d_G(A)\}$. This is a contradiction, since we have $d_H(A) < k$ and $d_H(A) < d_G(A)$.
\end{proof}

\begin{proof}[Proof of Lemma~\ref{l:iff}]
First suppose that $H$ is a feasible {\RTF} solution. We will show that both (i) and (ii) are satisfied.
Suppose to the contrary that (i) does not hold for some non-central component $J$. Let $H' = H[J]$. Let $u$ and $v$ be the attachment nodes of $J$, and assume without loss of generality that in $G$, any (simple) $st$-path has $u$ before $v$, as in Fig,~\ref{f:sp}(a).
Then, there exists a $J$-minimal $uv$-cut, $\de_{H'}(A_J)$, such that $d_{H'}(A_J) \leq 1$ (recall that $H'$ does not include any of the attachment edges of its associated $3$-class). 
Let $A$ be the union of $A_J$ and all components that lie on some (simple) $us$-path of the cycle chain that do not contain $v$. It can be seen that
$d_H(A)=d_{H'}(A_J)+1 \leq 2$. Since $u$ and $v$ are 3-connected in $G$ (they are in the same $3$-class) and $A$ is a $uv$-set, $d_G(A) \geq 3$. It is also not hard to see that $G[A]$ and $G[V \sem A]$ are both connected, and thus $\de_G(A)$ is a $G$-minimal $st$-cut. This all means that $F = \de_{H}(A)$ is a fault set of size at most 2 such that there is an $st$ path in $G \setminus F$, but no $st$ path in $H \setminus F$.
This contradicts the feasibility of $H$.

Now suppose that (ii) does not hold for some central component $J$. Then, there must be a $J$-minimal $st$-cut, $\de_{H'}(A_J)$, such that $d_{H'}(A_J) < \min\{3,d_J(A_J)\}$ (by Lemma~\ref{l:cover}).
Let $A$ be the union of $A_J$ and all components of cycles that precede $J$ in the cactus chain.
Note that $d_H(A)=d_{H'}(A_J)$ and $d_G(A)=d_J(A_J)$. This gives $d_H(A) < \min\{3,d_G(A)\}$, contradicting the feasibility of $H$.

Now suppose that $H$ satisfies (i,ii). We will show that $H$ is a feasible {\RTF} solution.
Let $F \subs E$ such that $G \sem F$ has an $st$-path and $H \sem F$ has no $st$-path.
One can see that since $F$ cannot contain two edges of the same cactus chain cycle, we have that $F$ must have size at least 3; thus, if there is an $st$ path in $G \setminus F$ for $|F| \leq 2$, then there is an $st$ path in $H \setminus F$ as well.
\end{proof}

\begin{lemma} \label{l:char}
Subgraph $H$ is a feasible solution for instance ${\cal I}$ if and only if for any $st$-set $A$ such that $d_H(A)=2$, $\de_H(A)=\de_J(s)$ or $\de_H(A)=\de_J(t)$.
\end{lemma}
\begin{proof}
Let $H$ be a feasible solution for ${\cal I}$ and let $A$ be an $st$-set with $d_H(A)=2$. We will show that either $\de_H(A)=\de_J(s)$ or $\de_H(A)=\de_J(t)$.
Let $F=\de_H(A)$. Then, $H \sem F$ has no $st$-path. Since $H$ is a feasible solution for ${\cal I}$, $J \sem F$ must also have no $st$-path. Component $J$ is $2$-connected, so $\de_H(A)=\de_J(A)$. By property (B) of Theorem~\ref{t:reduction}, we then have that $\de_J(A)=\de_J(s)$ or $\de_J(A)=\de_J(t)$.

Now suppose that for any $st$-set $A$ with $d_H(A)=2$, we have that $\de_H(A)=\de_J(s)$ or  $\de_H(A)=\de_J(t)$. We will show that $H$ is a feasible solution for ${\cal I}$. For any $st$-set $A$ such that $\de_H(A) \neq \de_J(s)$ and $\de_H(A) \neq \de_J(t)$, we have that $d_H(A) \geq 3$. Since $H$ includes the forced edges $sx,sy,zt,$ and $wt$, we also know that $\de_H(s)=\de_J(s)$ and $\de_H(t)=\de_J(t)$. Thus, we have that $d_H(A) \geq \min\{3,d_J(A)\} $ for all $A$ such that $A$ is an $st$-set. By Lemma~\ref{l:cover}, $H$ is a feasible solution for ${\cal I}$.
\end{proof}

\begin{figure}
\centering 
\includegraphics{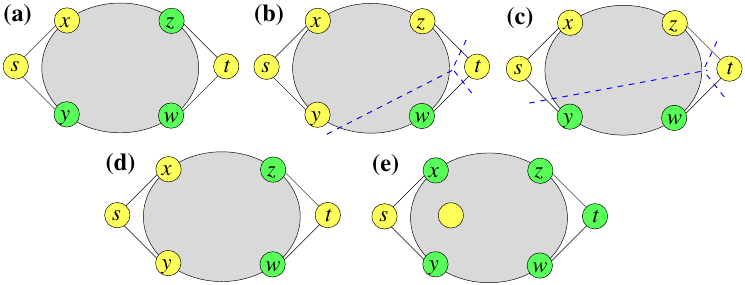}
\caption{Illustration to the proof of Lemmas \ref{l:3f}, \ref{l:3f'}. Nodes in $A$ are yellow and in $V \sem A$ are green.}
\label{f:3f}
\end{figure}

\begin{proof}[Proof of Lemma~\ref{l:3f}]
Suppose that $H$ is a feasible solution for ${\cal I}$, but suppose to the contrary that $R$ is not a $3$-connected subset in $H$. 
Then there is a vertex set $A$ such that $d_H(A) = 2$ and $R \cap A, R \sem A \neq \empt$.
By Lemma~\ref{l:char}, $A$ cannot be an $st$-set.
Hence, without loss of generality, assume that $s,t \in A$. Now we consider several cases, see Fig.~\ref{f:3f}(a,b,c,d).
\begin{enumerate}[(a)]
\item
$|R \cap A| =1$, say $R \cap A=\{x\}$;  see Fig.~\ref{f:3f}(a).
Then $\{sy,tz, tw\} \subs \de_H(A)$, contradicting that $d_H(A) = 2$.
\item
$|R \cap A|=3$, say $R \sem A = \{w\}$; see Fig.~\ref{f:3f}(b). Vertex set $A \sem \{t\}$ is an $st$-set with $d_H(A \sem \{t\})=d_H(A) = 2$, contradicting (by Lemma~\ref{l:char}) that $H$ is a feasible solution.
\item
$|R \cap A|=2$ and $|R \cap \{x,y\}|=|R \cap \{z,w\}|=1$, say $R \cap A=\{x,z\}$; see Fig.~\ref{f:3f}(c).
Then $d_H(A \sem \{t\})=d_H(A) = 2$, contradicting (by Lemma~\ref{l:char}) that $J$ is a feasible solution.
\item
$R \cap A=\{x,y\}$ or $R \cap A=\{z,w\}$, say $R \cap A=\{x,y\}$; see Fig.~\ref{f:3f}(d).
Then $A \sem \{t\}$ is an $st$-set with $d_H(A \sem \{t\}) = d_H(A)-2 = 0$, contradicting that $H$ is a feasible solution.
\end{enumerate}
In all cases we have a contradiction, thus the lemma holds.
\end{proof}

\begin{proof}[Proof of Lemma~\ref{l:3f'}]
Suppose that $R$ is a 3-connected subset in $H$, but suppose to the contrary that $H$ is not a feasible solution.
Then by Lemma~\ref{l:char} there is an $st$-set $A$ with $d_H(A)= 2$ such that $\de_H(A) \neq \de_J(s)$ and $\de_H(A) \neq \de_J(t)$. 
Since $R$ is a $3$-edge-connected subset in $H$, we must have $R \cap A=\empt$ or $R \subset A$. 
Without loss of generality, assume that $R \cap A=\empt$ (see Fig.~\ref{f:3f}(e)); otherwise we interchange the roles of $s$ and $t$. 
But then, $sx,sy \in \de_H(A)$, so we get the contradiction that $\de_H(A)=\{sx,sy\}=\de_J(s)$.
\end{proof}

\begin{proof}[Proof of Theorem~\ref{t:3SND}]
We reduce the $3$-RSND instance to an instance of SND. Let $D = \{d_1, d_2, \dots, d_{\ell} \}$ be the set of all $\ell$ vertex demands in the input $3$-RSND instance (where each $d_i \in D$ corresponds to a demand of the form $(x,y,r)$, where $x,y$ are vertices in $G$ and $r \leq 3$). We will create a new set of ordinary demands, $D'$, that characterizes the problem.

First note that $G$ is 2-connected, so all $1$- and $2$-demands in the input are ordinary demands. We add them to $D'$. We also add all ordinary $3$-demands from the input to $D'$. By Theorem~\ref{t:D}, each \emph{relative} $3$-demand $d_i \in D$ is characterized by a set of ordinary $3$-demands, $D_i$. Thus, for each relative $3$-demand $d_i \in D$, we add all demands in $D_i$ to $D'$. Demand set $D'$ characterizes the $3$-RSND instance, meaning that any optimal solution to the $3$-RSND instance must also satisfy $D'$. All demands in $D'$ are ordinary so we can run Jain's $2$-approximation SND algorithm (\cite{Jain01}) on $G$ under demands $D'$. By Theorem~\ref{t:D}, we only add a polynomial number of demands to $D'$, so the reduction (and the algorithm) runs in polynomial time. This gives us a $2$-approximation for $3$-RSND.
\end{proof}

\section{Proof of Theorem~\ref{thm:genk_characterize} (Structure Theorem)} \label{app:structure}

\subsection{Connectivity Lemmas for Input Graph G} 
To prove  Theorem~\ref{thm:genk_characterize}, we will first need lemmas that describe the connectivity in $G$ between $s$ and any $h$-component, and between any $h$-component and $t$. We will say that the \textit{level $h$ chain} is the ordered sequence of all $h$-components, labelled $R^h_0, \dots, R^h_p$ (where $s \in R^h_0$ and $t \in R^h_p$). Note that although all components in the level $h$ chain are $h$-components, each $h$-component may have a left or right separator with size less than $h$. Let $E(R^h_i)$ denote the set of edges with both endpoints in component $R^h_i$. We define the subchain of a level $h$ chain that starts at $G^h_i$ and ends at $G^h_j$ (where $i \leq j$) to be the subgraph of $G$ induced by $\cup^{j}_{k=i} R^h_k$. An edge $e$ is in a subchain that starts at $G^h_i$ and ends at $G^h_j$ if $e \in \cup^{j}_{k = i} E(R^h_k)$ or $e \in \cup^{j-1}_{k=i} S^h_k$. We define a subchain that starts at $H^h_i$ and ends at $H^h_j$ similarly. We may also use vertex sets $V^h_{(j, r)}$ to define the start and end of a subchain.

\begin{lemma}
\label{lem:s_connectivity_G}
Fix level $h$ of the hierarchical $k$-chain decomposition, and consider the subchain that starts at $s$ and ends at an $h$-component. Let $R^h_i$ be this $h$-component. Then, in $G$, there is a path from $s$ to each vertex in $V^h_{(i, r)}$; these paths only use edges within the level $h$ subchain that begins at $G^h_0$ and ends at $G^h_i$.
\end{lemma}
\begin{proof}
Graph $G$ is 2-edge connected, so the lemma statement is true when $R^h_i$ contains $t$ (that is, when $V^h_{(i, r)} = \{t\}$). Thus we will assume $t \notin R^h_i$.
First we will show that for each $h$-component $R^h_j$ such that $t \notin R^h_j$, there is a path in $G$ from $V^h_{(j, \ell)}$ to each vertex in $V^h_{(j, r)}$. Fix such an $h$-component $R^h_j$, and suppose that component $R^{h-1}_m$ is the parent of $R^h_j$ in the hierarchical $k$-chain decomposition. Suppose for the sake of contradiction that there is no path in $G$ from $V^h_{(j, \ell)}$ to a vertex $v \in V^h_{(j, r)}$.  Let $S_v$ be the set of edges in $S^h_j$ that are incident on $v$. Then, $S^h_j \setminus S_v$ is a $(V^h_{(j, \ell)}, V^{h-1}_{(m, r)})$-separator with size strictly less than $|S^h_j|$. This means that $S^h_j$ is not minimal and therefore not an important separator, giving a contradiction. 

We have shown that in each $h$-component $R^h_j$ such that $t \notin R^h_j$, there is a path in $G$ from $V^h_{(j, \ell)}$ to each vertex in $V^h_{(j, r)}$. Therefore, for every component $R^h_i$, there is a path in $G$ from $s$ to each vertex in $V^h_{(i, r)}$, using only edges within the subchain. This can be seen via a proof by induction on the number of components into the level $h$ subchain that begins at $R^h_0$ and ends at $R^h_i$.
\end{proof}

We now give a similar lemma describing the connectivity between any $h$-component and $t$.

\begin{lemma}
\label{lem:t_connectivity_G}
Fix level $h$ of the hierarchical $k$-chain decomposition, and consider the subchain that starts at an $h$-component and ends at $t$. Let $R^h_i$ be this $h$-component. Then, in $G$, there is a path from each vertex in $V^h_{(i, \ell)}$ to $t$; these paths only use edges within the level $h$ subchain that begins at $G^h_i$ and ends at $G^h_p$.
\end{lemma}
\begin{proof}
Graph $G$ is 2-edge connected, so the lemma statement is true when $R^h_i$ contains $s$ (that is, when $V^h_{(i, \ell)} = \{s\}$). Thus we will assume $s \notin R^h_i$. First we will show that for each $h$-component $R^h_j$ such that $s \notin R^h_j$, there is a path in $G$ from each vertex in $V^h_{(j, \ell)}$ to $V^h_{(j, r)}$. Fix such an $h$-component $R^h_j$, and suppose that component $R^{h-1}_m$ is the parent of $R^h_j$ in the hierarchical $k$-chain decomposition. Suppose for the sake of contradiction that there is no path in $G$ from a vertex $v \in V^h_{(j, \ell)}$ to $V^h_{(j, r)}$. Let $S_v$ be the set of edges in $S^h_{j-1}$ that are incident on $v$. Then, $S^h_{j-1} \setminus S_v $ is a $(V^h_{(j-1, \ell)}, V^{h-1}_{(m, r)})$-separator with size strictly less than $|S^h_{j-1}|$. This means that $S^h_{j-1}$ is not minimal and therefore not an important separator, giving a contradiction. 

We have shown that in each $h$-component $R^h_j$ such that $s \notin R^h_j$, there is a path in $G$ from each vertex in $V^h_{(j, \ell)}$ to $V^h_{(j, r)}$. Therefore, for every component $R^h_i$, there is a path in $G$ from $V^h_{(i, \ell)}$ to $t$, using only edges within the subchain. This can be seen via a proof by induction on the number of components into the level $h$ chain that beings at $R^h_i$ and ends at $R^h_p$.
\end{proof}

\subsection{Only If}
\label{sec:genk_only_if}
We are now ready to prove that the properties stated in Theorem~\ref{thm:genk_characterize} are necessary in any feasible solution. Let Property 1 denote the first demand set given in Theorem~\ref{thm:genk_characterize}---that is, the set $\big\{(X,Y,d) : X \subseteq V^h_{(i, \ell)}, Y \subseteq V^h_{(i, r)}, (X,Y) \neq (V^h_{(i, \ell)}, V^h_{(i, r)}),  d = \max(0,  k  + |S_X| + |S_Y| - |S^h_{i-1}| - |S^h_i|)  \big\}$. Let Property 2 denote the demand $(V^h_{(i, \ell)},V^h_{(i, r)},h+1)$ and Property 3 denote the demand $(V^h_{(i, \ell)},V^h_{(i, r)},k-1)$ in the union.

\subparagraph*{Important Separators.} Suppose that subgraph $H$ is a feasible solution, but suppose for the sake of contradiction that at level $h$ in the hierarchical $k$-chain decomposition, there is an important separator $S^h_i$ such that an edge $e \in S^h_i$ is not in $H$. Then, the edge set $E' = S^h_i \setminus \{e\}$, which has size at most $k-2$, separates $s$ and $t$ in $H$ but not in $G$. By Lemmas~\ref{lem:s_connectivity_G} and \ref{lem:t_connectivity_G}, in $G \setminus E'$, there is a path from $s$ to each vertex in $V^h_{(i,r)}$ and a path from each vertex in $V^h_{(i+1,\ell)}$ to $t$. The edge $e$ is incident on a vertex in $V^h_{(i,r)}$ and a vertex in $V^h_{(i+1,\ell)}$. Therefore, putting everything together, there is a path from $s$ to $t$ in $G \setminus E'$. In $H$ however, there is no path from $V^h_{(i,r)}$ to $V^h_{(i+1,\ell)}$, and therefore there is no path from $s$ to $t$ in $H \setminus E'$. This contradicts the assumption that $H$ is a feasible solution.

\subparagraph*{Property 1.} Now suppose $H$ is feasible, but suppose for the sake of contradiction that at least one of the RSND demands in Property 1 is not satisfied in subgraph $H^h_i$. Specifically, fix subgraph $H^h_i$, let $X$ and $Y$ be vertex sets such that $X \subseteq V^h_{(i,\ell)}$, $Y \subseteq V^h_{(i,r)}$, and $(X,Y) \neq (V^h_{(i, \ell)}, V^h_{(i, r)})$, and let $d = \max(0,  k  + |S_X| + |S_Y| - |S^h_{i-1}| - |S^h_i|)$. Suppose that the RSND demand $(X,Y,d)$ is not satisfied in $H^h_i$. That is, there exists a fault set $F_i$ with $|F_i| < d$ such that there is a path from $X$ to $Y$ in $G^h_i \setminus F_i$ but there is no path from $X$ to $Y$ in $H^h_i \setminus F_i$. Since we are assuming that the RSND demand $(X,Y,d)$ is not satisfied in $H^h_i$, we can assume that $d > 0$ (since a demand of 0 is always satisfied). Therefore we assume $d = k  + |S_X| + |S_Y| - |S^h_{i-1}| - |S^h_i| > 0$. 

We now show that there exists a fault set $F$ with $|F| < k$ such that $s$ and $t$ are connected in $G \setminus F$ but not in $H \setminus F$. Let $F = (S^h_{i-1} \setminus S_X) \cup (S^h_i \setminus S_Y) \cup F_i$. By Lemma \ref{lem:s_connectivity_G}, and by the fact that all edges in $S_X$ are still in $G \setminus F$, there is a path in $G \setminus F$ from $s$ to each vertex in $X$, using only edges before $G^h_i$ in the level $h$ chain. There is a path from $X$ to $Y$ in $G^h_i \setminus F_i$ (and in $G^h_i \setminus F$). By Lemma \ref{lem:t_connectivity_G}, and by the fact that all edges in $S_Y$ are still in $G \setminus F$, we also have that there is a path in $G \setminus F$ from each vertex in $Y$ to $t$ that only uses edges after $G^h_i$ in the level $h$ chain. Putting it all together, there is a path from $s$ to $t$ in $G \setminus F$. Note that because $X$ and $Y$ are not connected in $H \setminus F$, there is no $st$ path in $H \setminus F$. Now we give an upper bound on the size of $F$:
\begin{align*}
    |F| &\leq (|S^h_{i-1}| - |S_X|) + (|S^h_i| - |S_Y|) + (d-1) \\
    &= |S^h_{i-1}| - |S_X| + |S^h_i| - |S_Y| + k  + |S_X| + |S_Y| - |S^h_{i-1}| - |S^h_i| - 1 \\
    &= k-1.
\end{align*}
We have shown that $|F| \leq k-1$. Thus, our construction of $F$ contradicts the assumption that $H$ is feasible, so all RSND demands in the set $\big\{(X,Y,d) : X \subseteq V^h_{(i, \ell)}, Y \subseteq V^h_{(i, r)}, (X,Y) \neq (V^h_{(i, \ell)}, V^h_{(i, r)}),  d = \max(0,  k  + |S_X| + |S_Y| - |S^h_{i-1}| - |S^h_i|)  \big\}$ on $G^h_i$ are satisfied in $H^h_i$ in all feasible solutions.

\subparagraph*{Property 2.} Now, suppose $H$ is feasible, but suppose for the sake of contradiction that the RSND demand in Property 2, $(V^h_{(i, \ell)},V^h_{(i, r)},h+1)$, is not satisfied in a fixed subgraph $H^h_i$. That is, there exists a fault set $F_i$ with $|F_i| < h+1 \leq k$ such that there is a path from $V^h_{(i, \ell)}$ to $V^h_{(i, r)}$ in $G^h_i \setminus F_i$ but there is no path from $V^h_{(i, \ell)}$ to $V^h_{(i, r)}$ in $H^h_i \setminus F_i$. By Lemma~\ref{lem:s_connectivity_G}, there is a path from $s$ to each vertex in $V^h_{(i, \ell)}$ in $G \setminus F_i$, using only edges before $G^h_i$ in the level $h$ chain. There is a path from at least one vertex in $V^h_{(i, \ell)}$ to a vertex in $V^h_{(i, r)}$ in $G^h_i \setminus F_i$. By Lemma~\ref{lem:t_connectivity_G}, there is a path from each vertex in $V^h_{(i, r)}$ to $t$ in $G \setminus F_i$, using only edges after $G^h_i$ in the level $h$ chain. Putting it all together, there is a path from $s$ to $t$ in $G \setminus F$. However, there is no path from $V^h_{(i, \ell)}$ to $V^h_{(i, r)}$ in $H^h_i \setminus F_i$, so there is no $st$ path in $H \setminus F_i$. This contradicts the assumption that $H$ is feasible, so the RSND demand $(V^h_{(i, \ell)},V^h_{(i, r)},h+1)$ on $G^h_i$ must be satisfied in $H^h_i$ in all feasible solutions.

\subparagraph*{Property 3.} Finally, suppose $H$ is feasible, but suppose for the sake of contradiction that the RSND demand in Property 3, $(V^h_{(i, \ell)},V^h_{(i, r)},k-1)$, is not satisfied in a fixed subgraph $H^h_i$. Then there exists a fault set $F_i$ with $|F_i| < k-1$ such that there is a path from $V^h_{(i, \ell)}$ to $V^h_{(i, r)}$ in $G^h_i \setminus F_i$ but there is no path from $V^h_{(i, \ell)}$ to $V^h_{(i, r)}$ in $H^h_i \setminus F_i$. By Lemmas~\ref{lem:s_connectivity_G} and \ref{lem:t_connectivity_G}, there is a path from $s$ to $t$ in $G \setminus F_i$. However, there is no path from $V^h_{(i, \ell)}$ to $V^h_{(i, r)}$ in $H^h_i \setminus F_i$, so there is no $st$ path in $H \setminus F_i$. This contradicts the assumption that $H$ is feasible, so the RSND demand $(V^h_{(i, \ell)},V^h_{(i, r)},k-1)$ on $G^h_i$ must be satisfied in $H^h_i$ in all feasible solutions.

\subsection{If}
We now prove that the properties stated in Theorem~\ref{thm:genk_characterize} are sufficient. In this section, we will let $H$ be a subgraph of $G$ such that all edges in $S$ are in $H$ and all RSND demands in the statement of Theorem~\ref{thm:genk_characterize} are satisfied forall subgraphs $H^h_i$. For all possible fault sets $F \subseteq E$ such that $|F| < k$, we will show that if $s$ and $t$ are connected in $G \setminus F$, they must also be connected in $H \setminus F$. Going forward, fix $F \subseteq E$ to be a fault set such that $|F| < k$.

\subsubsection{Connectivity Lemmas for Subgraph H} We first give the equivalent of Lemmas~\ref{lem:s_connectivity_G} and \ref{lem:t_connectivity_G} for connectivity in $H$. The following lemmas describe the connectivity in $H$ between $s$ and any $h$-component, and between any $h$-component and $t$.

\begin{lemma}
\label{lem:s_connectivity_H}
Let $H$ be a subgraph of $G$, and suppose that all properties in the statement of Theorem~\ref{thm:genk_characterize} are satisfied by each subgraph $H^h_i$. Fix level $h$ of the hierarchical $k$-chain decomposition, and fix a subgraph $H^h_i$ in the level $h$ chain. Then, in $H$, there is a path from $s$ to each vertex in $V^h_{(i, r)}$; these paths only use edges within the level $h$ subchain that begins at $H^h_0$ and ends at $H^h_i$.
\end{lemma}
\begin{proof}
First we show that for any subgraph $H^h_j$, there is a nonzero RSND demand from the vertex set $V^h_{(j,\ell)}$ to each vertex $v \in V^h_{(j,r)}$ that is satisfied in $H^h_j$. That is, for all $v \in V^h_{(j,r)}$, the RSND demand $\big(V^h_{(j,\ell)}, v, d \big)$ with $d > 0$ on $G^h_j$ is satisfied in $H^h_j$. This would mean that there is a path in $H^h_j$ from $V^h_{(j,\ell)}$ to $v$ since there is a path in $G^h_j$ from $V^h_{(j,\ell)}$ to $v$ (from Lemma~\ref{lem:s_connectivity_G}).

We have two cases: If $\{v\} = V^h_{(j,r)}$, then the RSND demand $\big(V^h_{(j,\ell)}, v, h+1)$ on $G^h_j$ is a nonzero demand satisfied in $H^h_j$. Now suppose that $\{v\} \neq V^h_{(j,r)}$. To show that a nonzero RSND demand from $X = V^h_{(j,\ell)}$ to $Y = \{v\}$ is satisfied in $H^h_j$, we plug $|S_X| = |S^h_{j-1}|$ and $|S_Y| \geq 1$ into the formula $d = \max(0,  k  + |S_X| + |S_Y| - |S^h_{j-1}| - |S^h_{j}|)$. We plug in $|S_Y| \geq 1$ because $v$ is adjacent to at least one edge in $S^h_j$:
\begin{align*}
     d &= \max(0,  k  + |S_X| + |S_Y| - |S^h_{j-1}| - |S^h_{j}|) \\
     &\geq k  + |S_X| + |S_Y| - |S^h_{j-1}| - |S^h_{j}| \\
     &\geq k  + |S^h_{j-1}| + 1 - |S^h_{j-1}| - |S^h_{j}| \\
     &= k + 1  - |S^{h}_{j}| \\
     &> 0 \tag{$|S^h_j| \leq k-1$}.
\end{align*}
We have shown that there is a positive RSND demand from $V^h_{(j,\ell)}$ to each vertex in $V^h_{(j,r)}$. By Lemma \ref{lem:s_connectivity_G}, we have that in subgraph $G^h_j$, there is a path from $V^h_{(j,\ell)}$ to each vertex in $V^h_{(j,r)}$. Therefore, in $H^{h}_j$, there must be a path from $V^h_{(j,\ell)}$ to each vertex in $V^h_{(j,r)}$. This also implies that in $H$, each vertex in $V^h_{(i,r)}$ is reachable from $s$, using only edges within the subchain from $H^h_0$ to $H^h_i$. This can be seen via a proof by induction on the number of components into the level $h$ subchain that starts at $H^h_0$ and ends at $H^h_i$.
\end{proof} 

We now give a similar lemma describing the connectivity between any $h$-component and $t$. 

\begin{lemma}
\label{lem:t_connectivity_H}
Let $H$ be a subgraph of $G$, and suppose that all properties in the statement of Theorem~\ref{thm:genk_characterize} are satisfied by each subgraph $H^h_i$. Fix level $h$ of the hierarchical $k$-chain decomposition, and fix a subgraph $H^h_i$ in the level $h$ chain. Then, in $H$, there is a path from each vertex in $V^h_{(i, \ell)}$ to $t$; these paths only use edges within the level $h$ subchain that begins at $H^h_i$ and ends at $H^h_p$.
\end{lemma}
\begin{proof}
First we show that for all $v \in V^h_{(j,\ell)}$, the RSND demand $\big(v, V^h_{(j,r)} d \big)$ with $d > 0$ on $G^h_j$ is satisfied in $H^h_j$. This would mean that there is a path in $H^h_j$ from $v$ to $V^h_{(j,r)}$ since there is a path in $G^h_j$ from $v$ to $V^h_{(j,r)}$ (from Lemma~\ref{lem:t_connectivity_G}).

We have two cases: If $\{v\} = V^h_{(j,\ell)}$, then the RSND demand $\big(v, V^h_{(j,r)},  h+1)$ on $G^h_j$ is a nonzero RSND demand satisfied in $H^h_j$. To show that a nonzero RSND demand from $X = \{v\}$ to $Y = V^h_{(j,r)}$ on $G^h_j$ is satisfied in $H^h_j$, we plug $|S_X| \geq 1$ and $|S_Y| = |S^h_{j}|$ into the formula $d = \max(0,  k  + |S_X| + |S_Y| - |S^h_{j-1}| - |S^h_{j}|)$:
\begin{align*}
     d &= \max(0,  k  + |S_X| + |S_Y| - |S^h_{j-1}| - |S^h_{j}|) \\
     &\geq k  + |S_X| + |S_Y| - |S^h_{j-1}| - |S^h_j| \\
     &\geq k + 1 + |S^h_{j}| - |S^h_{j-1}| - |S^h_j| \\
     &= k + 1  - |S^{h}_{j-1}| \\
     &> 0. \tag{$|S^{h}_{j-1}| \leq k-1$} 
\end{align*}
We have shown that there is a positive RSND demand from each vertex in $V^h_{(j,\ell)}$ to $V^h_{(j,r)}$. By Lemma~\ref{lem:t_connectivity_G}, we have that in subgraph $G^h_j$, there is a path from each vertex in $V^h_{(j,\ell)}$ to $V^h_{(j,r)}$. Therefore, in $H^{h}_j$, there must be a path from each vertex in $V^h_{(j,\ell)}$ to $V^h_{(j,r)}$. This also implies that in $H$, there is a path from each vertex in $V^h_{(i,\ell)}$ to $t$, using only edges within the subchain from $H^h_i$ to $H^h_p$. This can be seen via a proof by induction on the number of components into the level $h$ subchain that starts at $H^h_i$ and ends at $H^h_p$.
\end{proof}  

\subsubsection{Analyzing Fault Sets via the Hierarchical Chain Decomposition}
We say that an edge $e$ is inside a subgraph $G^h_i$ in the hierarchical $k$-chain decomposition if $e \in E(R^h_i)$. That is, an edge is in a subgraph $G^h_i$ if the edge has both endpoints in $R^h_i$. Edges in the left or right separator of $G^h_i$ are not considered to be inside $G^h_i$. We now consider two cases for the locations of the edge faults in the hierarchical $k$-chain decomposition. In \textbf{Case 1}, all fault edges in $F$ are inside the same $(k-1)$-component in the decomposition. In \textbf{Case 2}, there is some level $h$ such that the edges in $F$ are not all inside the same $h$-component, but at level $h-1$, the edges in $F$ are inside a single $(h-1)$-component. 

\paragraph{Case 1.} We first consider Case 1, which is much simpler, and prove that there is an $st$ path in $H \setminus F$ under this case. Suppose we are in Case 1, and let $R^{k-1}_i$ be the $(k-1)$-component that contains all of $F$.

\begin{lemma}
\label{lem:case1}
Suppose that Case 1 applies in the hierarchical $k$-chain decomposition; that is, each edge in $F$ is inside the same $(k-1)$-component $R^{k-1}_i$. Suppose also that there is an $st$ path in $G \setminus F$. Then, there is an $st$ path in $H \setminus F$.
\end{lemma}
\begin{proof}
In $G^{k-1}_i$, there are at least $k$ edge-disjoint paths from $V^{k-1}_{(i,\ell)}$ to $V^{k-1}_{(i,r)}$ (otherwise, $G^{k-1}_i$ would have a $(V^{k-1}_{(i,\ell)}, V^{k-1}_{(i,r)})$-separator with size at most $k-1$, contradicting the structure of the hierarchical decomposition). In addition, the subgraph $H^{k-1}_i$ satisfies the RSND demand $(V^{k-1}_{(i,\ell)}, V^{k-1}_{(i,r)}, h+1)$ on $G^{k-1}_i$, where $h+1 = k$. Therefore, in $H^{k-1}_i$, there are also at least $k$ edge-disjoint paths from $V^{k-1}_{(i,\ell)}$ to $V^{k-1}_{(i,r)}$. Since $F$ has size at most $k-1$, there must be a path in $H^{k-1}_i \setminus F$ from $V^{k-1}_{(i,\ell)}$ to $V^{k-1}_{(i,r)}$. By Lemma~\ref{lem:s_connectivity_H}, there is a path from $s$ to each vertex in $V^{k-1}_{(i,\ell)}$, and by Lemma~\ref{lem:t_connectivity_H} there is a path from each vertex in $V^{k-1}_{(i,r)}$ to $t$. Combining all aforementioned paths gives a path from $s$ to $t$ in $H \setminus F$.
\end{proof}

\paragraph{Case 2.} We now consider Case 2, which is more complex. In Case 2 there is some level $h$ such that the edges in $F$ are not all inside the same $h$-component, but at level $h-1$, the edges in $F$ are inside a single $h-1$-component. Let $R^{h-1}_j$ be this $(h-1)$-component. Going forward, we will focus on what we will call the fault subchain of $G$.

\begin{definition}
The \textit{fault subchain} of $G$ is the $(V^{h-1}_{(j,\ell)}, V^{h-1}_{(j,r)})$ $h$-chain of $R^{h-1}_j$. Let $R_0, \dots, R_p$ be the $h$-components of the fault subchain, and let $S_{0}, \dots, S_{p-1}$ be the important separators adjacent to components in the fault subchain. Note that the left and right separators of $R^{h-1}_j$ are not included in the fault subchain.
\end{definition}

We will show that there is a path from $V_{(0,\ell)}$ to $V_{(p,r)}$ in the fault subchain after all faults are removed. With Lemmas~\ref{lem:s_connectivity_H} and \ref{lem:t_connectivity_H}, this will be enough to prove that there is an $st$ path in $H \setminus F$ given an $st$ path in $G \setminus F$. 

We first give some useful notation and definitions. We denote the prefix subchains of the fault subchain as follows. Let $L_i$ be all vertices in $\cup^i_{k=0} R_k$. Let $L_i^G$ be the subgraph of $G$ induced by $L_i$, and let $L_i^H$ be the subgraph of $H$ induced by $L_i$. Let $f_i$ be the number of edge faults in $L^G_i$.

Let $i \leq m$. In general, we will say that an edge $e \in S_m$ in the fault subchain is ``reachable'' from a vertex set $X \subseteq V_{(i,\ell)}$ in $L^G_m \setminus F$ (in $L^H_m \setminus F$) if there is a path in $L^G_m \setminus F$ (in $L^H_m \setminus F$) from $X$ to the vertex in $V_{(m,r)}$ that is incident on $e$, and the path only uses edges in the subchain from $G^h_i$ to $G^h_m$. Note that $e$ can be considered reachable even if $e \in F$. 
We will also use $S^G_i$ ($S^H_i$) to denote the set of edges in $S_i$ that are reachable from $V_{(0,\ell)}$ in $L^G_i \setminus F$ (in $L^H_i \setminus F$). That is, an edge $e$ is in $S^G_i$ (in $S^H_i$) if in $L^G_i \setminus F$ (in $L^H_i \setminus F$), there is a path from $V_{(0,\ell)}$ to the vertex in $V_{(i,r)}$ incident on $e$. In the next lemma, we prove a lower bound for the size of $S^G_i$.

\begin{lemma}[Lower bound on $|S^G_i|$]
\label{lem:G_F_reachable}
Suppose there is an $st$ path in $G \setminus F$. Let $L^G_i$ be a prefix subchain of the fault subchain such that $i \neq p$, and let $f_i$ be the number of edge faults in $L^G_i$. The edge set $S^G_i$ has size at least $h-f_i$.
\end{lemma}
\begin{proof}
Let $F_i \subseteq F$ be the set of fault edges in the subchain $L^G_i$, where $i \neq p$. Suppose for the sake of contradiction that $|S^G_i| < h-f_i$. Observe that the edge set $F_i \cup S^G_i$ separates $V_{(0,\ell)} = V^{h-1}_{(j,\ell)}$ and $V_{(p,r)} = V^{h-1}_{(j,r)}$. However, $F_i \cup S^G_i$ has size at most $f_i + (h-f_i-1) = h-1$. This means that $F_i \cup S^G_i$ is a $(V^{h-1}_{(j,\ell)}, V^{h-1}_{(j,r)})$-separator with size at most $h-1$. This is a contradiction, since $R^{h-1}_j$, which is an $(h-1)$-component, cannot have a $(V^{h-1}_{(j,\ell)}, V^{h-1}_{(j,r)})$-separator with size $h-1$ or less. We therefore have that $S^G_i$ has size at least $h-f_i$.
\end{proof}

For any subgraph $G^h_i$ in the fault subchain such that $0 < i < p$, we give a lower bound on the number of edge disjoint paths between any two vertex sets $X \subseteq V_{(i,\ell)}$ and $Y \subseteq V_{(i,r)}$.

\begin{lemma} [Edge-disjoint paths lower bound]
\label{lem:AFT}
Let $R_i$ be an $h$-component in the fault subchain such that $i \neq 0$ and $i \neq p$. Let $X \subseteq V_{(i,\ell)}$ and $Y \subseteq V_{(i,r)}$, and let $S_X$ and $S_Y$ be the edges in $S_{i-1}$ and in $S_{i}$ that are incident on vertices in $X$ and in $Y$, respectively. Let $d = \max(0, k + |S_X| + |S_Y| - |S_{i-1}|- |S_{i}|)$. In subgraph $G^h_i$, there are at least $d-k+h$ edge disjoint paths from $X$ to $Y$.
\end{lemma}
\begin{proof}
Suppose for the sake of contradiction that there are at most $d-k+h-1$ edge disjoint paths from $X$ to $Y$ in subgraph $G^h_i$, where $i \neq 0$ and $i \neq p$. Then there exists an edge set $E'$ with size at most $d-k+h-1$ such that $X$ and $Y$ are disconnected in $G^h_i \setminus E'$. Let $E_X = S_{i-1} \setminus S_X$ and let $E_Y = S_i \setminus S_Y$. Observe that $E' \cup E_X \cup E_Y$ separates $V_{(0,\ell)} = V^{h-1}_{(j,\ell)}$ and $V_{(p,r)} = V^{h-1}_{(j,r)}$. We give an upper bound on the number of edges in $E' \cup E_X \cup E_Y$:
\begin{align*}
    |E' \cup E_X \cup E_Y| &= |E'| + |E_X| + |E_Y| \\
    &\leq (d-k+h-1) + (|S_{i-1}| - |S_X|) + (|S_{i}| - |S_Y|) \\
    &= k + |S_X| + |S_Y| - |S_{i-1}|- |S_{i}| - k + h - 1 + |S_{i-1}| - |S_X| + |S_{i}| - |S_Y| \\
    &= h-1.
\end{align*}
We have shown that $E' \cup E_X \cup E_Y$ is a $(V^{h-1}_{(j,\ell)}, V^{h-1}_{(j,r)})$-separator with size at most $h-1$. This is a contradiction, since $R^{h-1}_j$ cannot have a $(V^{h-1}_{(j,\ell)}, V^{h-1}_{(j,r)})$-separator with size $h-1$ or less. We therefore have that there are at least $d-k+h$ edge disjoint paths from $X$ to $Y$.
\end{proof}

The following lemma is the crux of the argument for the sufficient condition. The lemma gives guarantees on the size of $S^H_i$ when $i \neq p$. Later, we will use this lemma to argue that at least one vertex in $V_{(p,r)}$ is reachable from $V_{(0,\ell)}$ in the fault subchain, after the removal of faults. Going forward, we will let $r_i = k-1-f_i$ be the maximum number of edge faults that are in $(\cup_{k=i}^{p-1} S_k) \cup (\cup^p_{k=i+1} E(R_k))$ (that is, the maximum number of edge faults that are not in $L^G_i$).

\begin{lemma}[Properties of $S^H_i$ when $i \neq p$]
\label{lem:r_induction}
Suppose there is an $st$ path in $G \setminus F$. Let $L^G_i$ be a prefix subchain such that $i \neq p$, and let $r_i = k-1-f_i$. Then, at least one of the following is true:
\begin{enumerate}
    \item $S_{i}^H = S_{i}^G$
    \item $|S_{i}^H| \geq r_i + 1$.
\end{enumerate}
\end{lemma}
\begin{proof}
We give a proof by induction on $i$, where $R_i$ is a component of the fault subchain.

\subparagraph*{Base Case.} Consider $R_0$ of the fault subchain. We will show that either $S_{0}^H = S_{0}^G$ or $|S_{0}^H| \geq r_0 + 1$. Let $F_R = F \cap E(R_0)$ be the set of fault edges with both endpoints in $R_0$, and let $f_R = |F_R|$. We will use $S_{-1}$ to denote the left separator of $V_{(0,\ell)}$. 

Using the RSND demands given in the Theorem~\ref{thm:genk_characterize} statement, we will show that for certain vertex subsets $Y \subseteq V_{(0,r)}$, there exists an RSND demand $(V_{(0,\ell)}, Y, d)$ on $G^h_0$ such that $d > f_R$. Using the RSND demands from $V_{(0,\ell)}$ to these subsets $Y \subseteq V_{(0,r)}$, we will show that $S_0^H = S_0^G$ or $|S_0^H| = r_0 + 1$. Recall that for a vertex set $Y \subseteq V_{(0,r)}$, $S_Y$ is the set of edges in $S_0$ that are incident on vertices in $Y$. We will first show that for all subsets $Y \subseteq V_{(0,r)}$ such that $|S_Y| \geq h - r_0$, there exists an RSND demand, $(V_{(0,\ell)}, Y, d)$, such that $d > f_R$. If $Y = V_{(0,r)}$, then the RSND demand $(X,Y,k-1)$ is satisfied in $H^h_0$, regardless of the value of $|S_Y|$. Since $f_R \leq k-2$ (recall that the edge faults are not all inside the same $h$-component in the fault subchain), we have that the RSND demand $(X,Y,f_R+1)$ is satisfied in $H^h_0$. Now we consider $Y$ such that $Y \neq V_{(0,r)}$. In this case, let $Y \subseteq V_{(0,r)}$ such that $|S_Y| \geq h - r_0$, and let $X = V_{(0,\ell)}$. Since the fault subchain is made up of at least two $h$-components, $|S_0| = h$. Plugging in $|S_X| = |S_{-1}|$, $|S_Y| \geq  h - r_0$, and $|S_0| = h$  into $ d = \max(0, k + |S_X| + |S_Y| - |S_{-1}| - |S_{0}|)$, we have the following:
\begin{align*}
    d &= \max(0, k + |S_X| + |S_Y| - |S_{-1}| - |S_{0}|) \\
    &\geq k + |S_X| + |S_Y| - |S_{-1}| - |S_{0}| \\
    &\geq k + (h - r_0) - (h) \tag{$|S_X| = |S_{-1}|$, $|S_Y| \geq  h - r_0$, $|S_0| = h$}\\ 
    &= k-r_0 \\
    &= k-k+1+f_0 \tag{$r_0 = k-1-f_0$}\\
    &= f_0+1 \\
    &\geq f_R+1. \tag{$f_0 \geq f_R$}
\end{align*}
We have that if $|S_Y| \geq h - r_0$, the RSND demand $(V_{(0,\ell)},Y,f_R+1)$ on $G^h_0$ is satisfied by $H^h_0$. Now we consider two cases and show that in the first case (1), $|S_{0}^H| \geq r_0 + 1$ and in the second case (2), $S_{0}^H = S_{0}^G$.
\begin{itemize}
    \item \textbf{(1) Suppose first that $|S_{0}^G| > r_0$.} This means that $S_0 \setminus S_{0}^G$, the set of edges in $S_0$ not reachable from $V_{(0,\ell)}$ in $G^h_0 \setminus F_R$, has size at most $|S_{0}| - r_0 - 1 = h - r_0 - 1$. Let $Y \subseteq V_{(0,r)}$ such that $|S_Y| \geq h - r_0$. Then, $S_Y$ must contain at least one edge in $S_{0}^G$, so at least one vertex in $Y$ is incident on an edge in $S_{0}^G$. Thus, there is a path from $V_{(0,\ell)}$ to $Y$ in $G^h_0 \setminus F_R$. We have shown that for all $Y \subseteq V_{(0,r)}$ such that $|S_Y| \geq h - r_0$, there is a path in $G^h_0 \setminus F$ from $V_{(0,\ell)}$ to $Y$. We also have that the RSND demands $\{(V_{(0,\ell)},Y,f_R+1) : Y \subseteq V_{(0,r)}, |S_Y| \geq h - r_0  \}$ on $G^h_0$ are satisfied in $H^h_0$. Therefore, in $H^h_0 \setminus F_R$, there is also a path from $V_{(0,\ell)}$ to a set $Y \subseteq V_{(0,r)}$ if $|S_Y| \geq h - r_0$.
    
    Suppose for the sake of contradiction that $|S^H_0| < r_0+1$. This means that $S_0 \setminus S_{0}^H$ has size at least $h - r_0$. The set $S_0 \setminus S_{0}^H$ is not reachable from $V_{(0,\ell)}$ in $H^h_0 \setminus F_R$. However, since $|S_0 \setminus S_{0}^H| \geq h - r_0$, there must be a path in $H^h_0 \setminus F_R$ from $X$ to $U$, where $U$ is set of vertices incident on at least one edge in $S_0 \setminus S_{0}^H$. This is a contradiction, so we have that $S^H_0$ must have size at least $r_0+1$.
    
    \item \textbf{(2) Now suppose that $|S_{0}^G| \leq r_0$.} This means that $S_0 \setminus S_{0}^G$ has size at least $h - r_0$. Let $e$ be an edge in $S_0^G$. Additionally, let $E_U \subseteq S_0 \setminus S_{0}^G$ such that $E_U$ has size at least $h - r_0 - 1$. Let $S_{e} = \{e\} \cup E_U$. Let $Y_{e}$ be the set of vertices in $V_{(0,r)}$ that are incident on at least one edge in $S_{e}$, and let $v_e$ be the vertex in $Y_e$ incident on $e$. Let $S_{Y_e}$ be the set of edges in $S_0$ that are incident on a vertex in $Y_e$ (note that $S_e \subseteq S_{Y_e}$). Since $S_{e}$ has size at least $h - r_0$ and $|S_{Y_e}| \geq |S_e|$, the RSND demand $(V_{(0,\ell)},Y_{e},f_R+1)$ on $G^h_0$ is satisfied in $H^h_0$. Edge $e$ is the only edge in $S_e$ that is reachable from $V_{(0,\ell)}$ in $G^h_0 \setminus F_R$, which means that $v_e$ is the only vertex in $Y_e$ that is reachable from $V_{(0,\ell)}$ in $G^h_0 \setminus F_R$. The RSND demand $(V_{(0,\ell)},Y_{e},f_R+1)$ is satisfied in $H^h_0$. The only way for this demand to be satisfied is if there is a path from $V_{(0,\ell)}$ to $v_e$ (and thus to $e$) in $H^h_0 \setminus F_R$. Applying this argument to all $e \in S^G_0$, in $H^h_0$ there must be a path from $V_{(0,\ell)}$ to each edge in $S^G_0$, so $S_0^H = S_0^G$. 
\end{itemize}
We have shown that if $|S^G_0| \geq r_0+1$ then $|S^H_0| \geq r_0+1$, and that if $|S^G_0| \leq r_0$ then $S^H_0 = S^G_0$.

\subparagraph*{Inductive Step.} We will now show that when $i \neq p$, if the lemma is true for $R_{i-1}$, then it also holds for $R_i$.  Let $F_S = F \cap S_{i-1}$ be the set of fault edges in the left separator of $R_i$. Let $F_R = F \cap E(R_i)$, $f_S = |F_S|$ and $f_R = |F_R|$. We also let $X$ be the set of vertices in $V_{(i,\ell)}$ that are incident on at least one edge in $S^H_{i-1} \setminus F_S$. 
A vertex $v \in V_{(i,\ell)}$ is in $X$ if there is a path from $V_{(0,\ell)}$ to $v$ in $G \setminus F$ that only uses edges in $L^G_{i-1} \cup S_{i-1}$. We also let $S_X$ be the set of edges in $S_{i-1}$ incident on a vertex in $X$. Note that $S_{i-1} \setminus F_S \subseteq S_X$. Recall that $f_{i-1}$ is the number of faults in $L^G_{i-1}$, and that $r_{i-1} = k-1-f_{i-1}$. We have two cases by the inductive hypothesis: (1) $S^H_{i-1} = S^G_{i-1}$ or (2) $|S^H_{i-1}| \geq r_{i-1} + 1$.

\begin{itemize}
    \item \textbf{(1) Suppose first that $S^H_{i-1} = S^G_{i-1}$.} As in the base case, we first show that for all $Y \subseteq V_{(i,r)}$ such that $|S_Y| \geq h - r_i$, there exists an RSND demand $(X,Y,d)$, with $d > f_R$, that is satisfied in $H^h_i$. Again, if $(X,Y) = (V_{(i,\ell)}, V_{(i,r)})$, then the RSND demand $(X,Y,f_R+1)$ is satisfied in $H^h_i$. Now we consider $X$ and $Y$ such that $(X,Y) \neq (V_{(i,\ell)}, V_{(i,r)})$. In this case, let $Y \subseteq V_{(i,r)}$ such that $|S_Y| \geq h - r_i$. Since $i-1 \neq p$, by Lemma~\ref{lem:G_F_reachable} we have that $|S^G_{i-1}| \geq h - f_{i-1}$. Therefore, $|S_X| \geq |S^H_{i-1} \setminus F_S| = |S^G_{i-1} \setminus F_S| \geq |S^G_{i-1}| - f_S \geq h - f_{i-1} - f_S$. Also note that since $0 < i < p$, we have that $|S_{i-1}|=|S_i| = h$.
    Plugging in  $|S_X| \geq h - f_{i-1} - f_S$, $|S_Y| \geq h - r_i$, and $|S_{i-1}| = |S_{i}| = h$ into $d = \max(0, k + |S_X| + |S_Y| - |S_{i-1}| - |S_{i}|)$, we get the following:
    \begin{align*}
        d &= \max(0, k + |S_X| + |S_Y| - |S_{i-1}| - |S_{i}|) \\
        &\geq k + |S_X| + |S_Y| - |S_{i-1}| - |S_{i}| \\
        &\geq k  + (h - f_{i-1} - f_S) + (h - r_i) - |S_{i-1}| - |S_{i}| \tag{$|S_X| \geq h - f_{i-1} - f_S$, $|S_Y| \geq  h - r_i$} \\
        &= k  + h - f_{i-1} - f_S + h - r_i - (h) - (h) \tag{$|S_{i-1}| = |S_{i}| = h$} \\
        &= k - f_{i-1} - f_S - r_i    \\
        &= k- f_{i-1} - f_S - k + 1 + f_{i-1} + f_S + f_R    \tag{$r_i = k-1-f_{i-1}-f_S-f_R$} \\
        &= f_R+1.
    \end{align*}
    We have shown that if $Y \subseteq V_{(i,r)}$ and $|S_Y| \geq h - r_i$, the RSND demand $(X,Y,f_R+1)$ on $G^h_i$ is satisfied by $H^h_i$. The rest of the argument is identical to that of the Base Case.
    
    \item \textbf{(2) Now suppose that $|S^H_{i-1}| \geq r_{i-1} + 1$.} We will first show that for all $Y \subseteq V_{(i,r)}$ such that $|S_Y| \geq h - r_i$, there are at least $f_R+1$ edge-disjoint paths in $G^h_i$ from $X$ to $Y$. Then, as in the previous case, we will show that for all $Y \subseteq V_{(i,r)}$ such that $|S_Y| \geq h - r_i$, the RSND demand $(X,Y,f_R+1)$ on $G^h_i$ is satisfied in $H^h_i$. Since $0<i<p$, by Lemma~\ref{lem:AFT} there are at least $d-k+h$ edge-disjoint paths from $X \subseteq V_{(i,\ell)}$ to $Y \subseteq V_{(i,r)}$ in $G^h_i$, where $d = \max(0, k + |S_X| + |S_Y| - |S_{i-1}| - |S_{i}|)$. Additionally, $|S_X| \geq |S^H_{i-1}| - f_S \geq r_{i-1} + 1 - f_S$. Plugging in  $|S_X| \geq r_{i-1} + 1 - f_S$, $|S_Y| \geq h - r_i$, and $|S_{i-1}|=|S_i| = h$ into $d = \max(0, k + |S_X| + |S_Y| - |S_{i-1}| - |S_{i}|)$, we get the following:
    \begin{align*}
        d-k+h &= \max(0, k + |S_X| + |S_Y| - |S_{i-1}| - |S_{i}|) - k + h \\
        &\geq k + |S_X| + |S_Y| - |S_{i-1}| - |S_{i}| - k + h \\
        &\geq (r_{i-1} + 1 - f_S) + (h - r_i)  - |S_{i-1}| - |S_{i}| + h \tag{$|S_X| \geq r_{i-1} + 1 - f_S$, $|S_Y| \geq  h - r_i$} \\
        &= r_{i-1} + 1 - f_S + h - r_i  - (h) - (h) + h \tag{$|S_{i-1}| = |S_i| = h$} \\
        &= r_{i-1} + 1 - f_S - r_i  \\
        &= (r_{i}+f_R+f_S) + 1 - f_S - r_i    \tag{$r_{i-1} = r_{i}+f_R+f_S$} \\
        &= f_R+1.
    \end{align*}
    We have shown that if $Y \subseteq V_{(i,r)}$ and $|S_Y| \geq h - r_i$, then $d-k+h \geq f_R+1$. This means that in $G^h_i$ there are at least $f_R+1$ edge disjoint paths from $X$ to $Y$ if $Y \subseteq V_{(i,r)}$ and $|S_Y| \geq h - r_i$. Therefore, in $G_i^h \setminus F_R$, there is a path from $X$ to all $Y \subseteq V_{(i,r)}$ such that $|S_Y| \geq h - r_i$. As before, if $(X,Y) = (V_{(i,\ell)}, V_{(i,r)})$, then the RSND demand $(X,Y,f_R+1)$ is satisfied in $H^h_i$. We now show that this demand is also satisfied when $(X,Y) \neq (V_{(i,\ell)}, V_{(i,r)})$:
    \begin{align*}
        d-k+h \geq f_R+1 \implies d &\geq f_R + 1 + k - h \\
        &\geq f_R + 1.  \tag{$k-h >  0$}
    \end{align*}
    Thus, if $Y \subseteq V_{(i,r)}$ with $|S_Y| \geq h - r_i$, then the RSND demand $(X,Y,f_R+1)$ on $G^h_i$ is satisfied by $H^h_i$. This, combined with the fact that in $G_i^h \setminus F_R$ there is a path from $X$ to all $Y \subseteq V_{(i,r)}$ such that $|S_Y| \geq h - r_i$, implies that there is a path in $H^h_i \setminus F_R$ from $X$ to all $Y \subseteq V_{(i,r)}$ such that $|S_Y| \geq h - r_i$. Now, suppose for the sake of contradiction that $|S^H_i| \leq r_i$. This means that $S_i \setminus S^H_i$ has size at least $h-r_i$. The set $S_i \setminus S^H_i$ is not reachable from $X$ in $H^h_i \setminus F_R$.  However, since $|S_i \setminus S^H_i| \geq h-r_i$, there must be a path in $H^h_i \setminus F$ from $X$ to $U$, where $U$ is the set of vertices incident on at least one edge in $S_i \setminus S^H_i$. This is a contradiction. We therefore have that $|S^H_i| \geq r_i+1$.
\end{itemize}

We have shown that for all $i$ such that $i \neq p$, either $S^H_i = S^G_i$ or $|S^H_i| \geq r_i+1$. 
\end{proof}

Now we use Lemma~\ref{lem:r_induction} to show that there is a path in the fault subchain from $V_{(0,r)}$ to at least one vertex in $V_{(p,r)}$. The proof uses arguments similar to those in Lemma~\ref{lem:r_induction}.

\begin{lemma}
\label{lem:the_end}
Suppose there is an $st$ path in $G \setminus F$. There is a path from $V_{(0,r)}$ to $V_{(p,r)}$ in $L^H_p \setminus F$.
\end{lemma}
\begin{proof}
Note that since $F$ is not contained in a single $h$-component, the fault subchain has at least two components, so component $R_{p-1}$ must exist. Let $F_S = F \cap S_{p-1}$, $F_R = F \cap E(R_p)$, $f_S = |F_S|$, and $f_R = |F_R|$. Let $X$ be the set of vertices in $V_{(p,\ell)}$ that are incident on at least one edge in $S^H_{p-1} \setminus F_S$. In addition, let $S_X$ be the set of edges in $S_{p-1}$ incident on a vertex in $X$, and let $S_p$ denote the right separator of $R_p$. By Lemma~\ref{lem:r_induction}, one of the following holds: (1) $|S^H_{p-1}| = |S^G_{p-1}|$ or (2) $|S^H_{p-1}| \geq r_{p-1} + 1$.
\begin{itemize}
    \item \textbf{(1) Suppose first that $|S^H_{p-1}| = |S^G_{p-1}|$.} There is an $st$ path in $G \setminus F$; therefore, there is a path in $G^h_p \setminus F_R$ from $X$ to $V_{(p,r)}$. We now show that the RSND demand $(X,V_{(p,r)},f_R+1)$ on $G^h_p$ is satisfied in $H^h_p$. If $X = V_{(p,\ell)}$, the RSND demand $(X,V_{(p,r)},f_R+1)$ is satisfied. Now suppose $X \neq V_{(p,\ell)}$. By Lemma~\ref{lem:G_F_reachable}, we have that $|S^G_{p-1}| \geq h-f_{p-1}$. We therefore have that $|S_X| \geq |S^H_{p-1} \setminus F_S| = |S^G_{p-1} \setminus F_S| \geq h - f_{p-1} - f_S$. Let $Y = V_{(p,r)}$. We plug in $|S_X| \geq h - f_{p-1} - f_S$, $|S_Y| = |S_{p}|$, and $|S_{p-1}| = h$ into $d = \max(0, k + |S_X| + |S_Y| - |S_{p-1}| - |S_{p}|)$:
    \begin{align*}
        d &= \max(0, k + |S_X| + |S_Y| - |S_{p-1}| - |S_{p}|) \\
        &\geq k + |S_X| + |S_Y| - |S_{p-1}| - |S_{p}| \\
        &\geq k + (h - f_{p-1} - f_S) - (h)  \tag{$|S_X| \geq h - f_{p-1} - f_S$, $|S_Y| = |S_{p}|$, $|S_{p-1}| = h$} \\
        &= k - f_{p-1} - f_S   \\
        &\geq k - k + 1 + f_S + f_R -f_S  \tag{$f_{p-1} \leq k-1-f_S-f_R$} \\
        &= f_R + 1. 
    \end{align*}
    We have shown that the RSND demand $(X,V_{(p,r)},f_R+1)$ on $G^h_p$ is satisfied in $H^h_p$. Since there is a path in $G^h_p \setminus F_R$ from $X$ to $V_{(p,r)}$, this means there must also be a path in $H^h_p \setminus F_R$ from $X$ to $V_{(p,r)}$. This, with Lemma~\ref{lem:r_induction}, implies a path from $V_{(0,\ell)}$ to $V_{(p,r)}$.
    
    \item \textbf{(2) Now suppose that $|S^H_{p-1}| \geq r_{p-1} + 1$.} This means that $|S_X| \geq r_{p-1} + 1-f_S$. We will first show that in $G^h_p$, there are at least $f_R+1$ edge-disjoint paths from $X$ to $V_{(p,r)}$. Suppose for the sake of contradiction that there are at most $f_R$ edge-disjoint paths $X$ to $V_{(p,r)}$. Then there exists a set $E'$ size at most $f_R$ such that $X$ and $V_{(p,r)}$ and not connected in $G^h_p \setminus E'$. Let $E_X = S_{p-1} \setminus S_X$. Observe that $E' \cup E_X$ separates $V_{(0,\ell)} = V^{h-1}_{(j,\ell)}$ and $V_{(p,r)} = V^{h-1}_{(j,r)}$. We give an upper bound on the number of edges in $E' \cup E_X$:
    \begin{align*}
        |E' \cup E_X| &= |E'| + |E_X| \\
        &\leq f_R + (|S_{p-1}| - |S_X|) \\
        &\leq f_R + h - r_{p-1} - 1 + f_S  \tag{$|S_{p-1}| = h$, $|S_X| \geq r_{p-1} + 1 - f_S$}  \\
        &\leq f_R + h - f_R - f_S - 1 + f_S \tag{$r_{p-1} \geq f_R+f_S$} \\
        &= h-1.
    \end{align*}
    We have shown that $E' \cup E_X$ is a $(V^{h-1}_{(j,\ell)}, V^{h-1}_{(j,r)})$-separator with size at most $h-1$. This is a contradiction, since $R^{h-1}_j$ cannot have a $(V^{h-1}_{(j,\ell)}, V^{h-1}_{(j,r)})$-separator with size $h-1$ or less. We therefore have that there are at least $f_R+1$ edge disjoint paths from $X$ to $V_{(p,r)}$ in $G^h_p$. This means that in $G_p^h \setminus F_R$, there is a path from $X$ to $V_{(p,r)}$. We now show that the RSND demand $(X,V_{(p,r)},f_R+1)$ is satisfied in $H^h_p$. As before, if $X = V_{(p,\ell)}$, then the demand is satisfied. Suppose $X \neq V_{(p,\ell)}$. Plugging in  $|S_X| \geq r_{p-1} + 1 - f_S$, $|S_Y| = |S_p|$, and $|S_{p-1}| = h$ into $d = \max(0, k + |S_X| + |S_Y| - |S_{p-1}| - |S_{p}|)$, we get the following:
    \begin{align*}
        d &= \max(0, k + |S_X| + |S_Y| - |S_{p-1}| - |S_{p}|) \\
        &\geq k + |S_X| + |S_Y| - |S_{p-1}| - |S_{p}|  \\
        &\geq k+ (r_{p-1} + 1 - f_S) - h  \tag{$|S_X| \geq r_{p-1} + 1 - f_S$, $|S_Y| = |S_p|$, $|S_{p-1}| = h$} \\
        &\geq k+ (f_R+f_S) + 1 - f_S - h \tag{$r_{p-1} \geq f_R+f_S$}  \\
        &= f_R + 1 + k - h  \\
        &\geq f_R+1. \tag{$k-h > 0$}
    \end{align*}
    We have shown that the RSND demand $(X,V_{(p,r)},f_R+1)$ on $G^h_p$ is satisfied by $H^h_p$. Since there is a path in $G^h_p \setminus F_R$ from $X$ to $V_{(p,r)}$, this means that there must also be a path in $H^h_p \setminus F_R$ from $X$ to $V_{(p,r)}$. This, with Lemma~\ref{lem:r_induction}, implies there is a path from $V_{(0,\ell)}$ to $V_{(p,r)}$.
\end{itemize}
\end{proof}

\begin{corollary}
\label{cor:end}
Suppose that Case 2 applies in the hierarchical $k$-chain decomposition; that is, there is some level $h$ such that the edges in $F$ are not all inside the same $h$-component, but at level $h-1$, the edges in $F$ are inside a single $(h-1)$-component. Suppose there is an $st$ path in $G \setminus F$. Then, there is an $st$ path in $H \setminus F$.
\end{corollary} 
\begin{proof}
Suppose neither $s$ nor $t$ is in the fault subchain. By Lemma~\ref{lem:s_connectivity_H}, there is a path in $H \setminus F$ from $s$ to each vertex in $V_{(0,\ell)}$, using only edges in the level $h$ subchain starting at $s$ and ending at $V_{(0,\ell)}$. Since there must be a path from $V_{(0,\ell)}$ to $V_{(p,r)}$ in $L^G_p \setminus F$, by Lemma~\ref{lem:the_end}, there is also a path from $V_{(0,\ell)}$ to $V_{(p,r)}$ in $L^H_p \setminus F$. Finally, by Lemma~\ref{lem:t_connectivity_H}, there is a path in $H \setminus F$ from each vertex in $V_{(p,r)}$ to $t$, using only edges in the level $h$ subchain starting at $V_{(p,r)}$ and ending at $t$. Combining all aforementioned paths gives a path in $H \setminus F$ from $s$ to $t$. If $s$ or $t$ is in the fault subchain, then a similar argument shows that there is an $st$ path in $H \setminus F$.
\end{proof}

With both Lemma~\ref{lem:case1} (Case 1) and Corollary~\ref{cor:end} (Case 2), we have shown that if $H$ satisfies the set of demands specified in Theorem~\ref{thm:genk_characterize}, then $H$ is a feasible solution to SD-$k$-RSND.

\section{Proof of Lemma~~\ref{lem:SD_feasible} } %from Section~\ref{sec:SD_cost}
\label{app:cost}
\begin{proof}[Proof of Lemma~\ref{lem:SD_feasible}]
Fix level $h$ of the hierarchical $k$-chain decomposition. For each $i$, let $H^h_i$ denote the subgraph of $H$ induced by $R^h_i$, and let $G^h_i$ denote the subgraph of $G$ induced by $R^h_i$. We will show that $H$ satisfies the conditions stated in Theorem~\ref{thm:genk_characterize}, and hence is feasible. By construction, $H$ contains $S$, or all edges in the important separators. First we show that the set of RSND demands $\big\{(X,Y,d) : X \subseteq V^h_{(i, \ell)}, Y \subseteq V^h_{(i, r)}, (X,Y) \neq (V^h_{(i, \ell)}, V^h_{(i, r)}), d = \max(0,  k  + |S_X| + |S_Y| - |S^h_{i-1}| - |S^h_i|)  \big\}$ on $G^h_i$ are satisfied in $H^h_i$.

Let $X \subseteq V^h_{(i,\ell)}$, $Y \subseteq V^h_{(i,r)}$, and $d = \max(0, k + |S_X| + |S_Y| - |S_{i-1}| - |S_{i}|)$, as in Theorem~\ref{thm:genk_characterize}.  First we show that for all $X,Y$ pairs such that $d = 1$, the RSND demand $(X,Y,d)$ on $G^h_i$ is satisfied in $H^h_i$. If there is a path from $X$ to $Y$ in $G^h_i$, we include in $H^h_i$ a path from $X$ to $Y$. Therefore, the demand $(X,Y,1)$ is satisfied.

Now we show that for all $X,Y$ pairs such that $1 < d < k$, the RSND demand $(X,Y,d)$ on $G^h_i$ is satisfied in $H^h_i$. We include in $H^h_i$ all edges selected via the SD-$d$-RSND algorithm on contracted $X$ and contracted $Y$ with demand $d$. Therefore, the demand $(X,Y,d)$ is satisfied. We have shown that if $d < k$, then the RSND demand $(X,Y,d)$ is satisfied. Since the demand for all $X,Y$ pairs such that $(X,Y) \neq (V^h_{(i,\ell)}, V^h_{(i,r)})$ is at most $k-1$, we have shown that the set of demands $\big\{(X,Y,d) : X \subseteq V^h_{(i, \ell)}, Y \subseteq V^h_{(i, r)}, (X,Y) \neq (V^h_{(i, \ell)}, V^h_{(i, r)}), d = \max(0,  k  + |S_X| + |S_Y| - |S^h_{i-1}| - |S^h_i|)  \big\}$ on $G^h_i$ is satisfied in $H^h_i$.

Now we show that the RSND demand $\big(V^h_{(i, \ell)}, V^h_{(i, r)}, k-1 \big)$ on $G^h_i$ is satisfied in $H^h_i$. We include in $H^h_i$ all edges selected via the SD-$(k-1)$-RSND algorithm on contracted $V^h_{(i, \ell)}$ and contracted $V^h_{(i, r)}$ with demand $k-1$. Therefore, the demand $(V^h_{(i, \ell)}, V^h_{(i, r)},k-1)$ is satisfied.

Finally we show that the RSND demand $\big(V^h_{(i, \ell)}, V^h_{(i, r)}, h+1 \big)$ on $G^h_i$ is satisfied in $H^h_i$. Due to the construction of the hierarchical $k$-chain decomposition, there are at least $h+1$ edge-disjoint paths from $V^h_{(i, \ell)}$ to $V^h_{(i, r)}$ in $G^h_i$. Thus, to show that the RSND demand is satisfied in $H^h_i$, it suffices to show that there are $h+1$ edge-disjoint paths from $V^h_{(i, \ell)}$ to $V^h_{(i, r)}$ in $H^h_i$. In $H^h_i$, we include the edges selected via a Min-Cost Flow algorithm from contracted $V^h_{(i, \ell)}$ to contracted $V^h_{(i, r)}$ with flow $h+1$. Since there are $h+1$ edge-disjoint paths from $V^h_{(i, \ell)}$ to $V^h_{(i, r)}$ in $G^h_i$, the Min-Cost Flow instance must have a feasible solution, and the algorithm will return $h+1$ edge-disjoint paths from $V^h_{(i, \ell)}$ to $V^h_{(i, r)}$ (i.e. the set of edges with non-zero flow).
\end{proof}

%%%%%%%%%%%%%%%%%%%%
\section{Simplification of {\em k}-EFTS 2-approximation} \label{app:efts}
%%%%%%%%%%%%%%%%%%%%
As discussed, a $2$-approximation for $k$-EFTS was given by~\cite{ap22}.  They showed that while the ``cut-requirement'' function defined by $k$-EFTS was not weakly supermodular (so the algorithm of~\cite{Jain01} could not be applied in a black-box way), it did have a property they dubbed ``local weak supermodularity'' which was enough to be able to use~\cite{Jain01}.  However, this was relatively involved, and required a careful case analysis of the cut requirement function.  We show here that more standard techniques can be used to achieve the same results in a far more simple manner.

\subsection{Weakly \texorpdfstring{$\FF$}{F}-Supermodular Definition and Properties}
First we give some preliminary definitions. Two sets $A,B$ {\bf cross} if $A \cap B \neq \empt$ and $A \cup B \neq V$, and 
{\bf co-cross} if $A \sem B,B \sem A \neq \empt$. 
Given $A,B \in \FF$, we say that $f$ satisfies the:
\begin{itemize}
\item
{\bf supermodular inequality} if $f(A)+f(B) \geq f(A \cap B)+f(A \cup B)$.
\item
{\bf co-supermodular inequality} if $f(A)+f(B) \geq f(A \sem B)+f(B \sem A)$.
\end{itemize}
A set function $f$ is {\bf symmetric} if $f(A)=f(V \sem A)$ for all $A \subs V$, and $f$ is: 
\begin{itemize}
\item
{\bf crossing supermodular} if $f$ satisfies the supermodular inequality whenever $A,B$ cross.
\item
{\bf weakly supermodular} if $f$ satisfies the supermodular or 
the co-supermodular inequality for all $A,B \subs V$.
\end{itemize}
Note that if $f$ is symmetric crossing supermodular and $A,B$ co-cross, 
then (by symmetry) $f$ satisfies the co-supermodular inequality. 

Given a set family $\FF$ with $\empt, V \notin \FF$, we may want to consider the restriction $f_{\FF}$ of $f$ to sets in $\FF$. 
We say that $f$ is {\bf weakly $\FF$-supermodular} if $f_{\FF}$ is weakly supermodular, namely, 
if for all $A,B \in \FF$ at least one of the following holds:
\begin{itemize}
\item
$A \cap B,A \cup B \in \FF$ and $f$ satisfies the supermodular inequality.
\item
$A \sem B,B \sem A \in \FF$ and $f$ satisfies the co-supermodular inequality. 
\end{itemize}
From Jain's result it follows that the problem of covering  a weakly $\FF$-supermodular function 
admits approximation ratio $2$ (assuming we can solve certain LPs in polynomial time).

\begin{lemma} \label{l:f}
Let $G=(V,E)$ be a graph, let $I \subs E$, let 
$f$ be a symmetric crossing supermodular set function on $2^V \sem \{\empt,V\}$. 
Then $f$ is weakly $\FF$-supermodular for $\FF=\{A \subs V:d_I (A) \geq 1\}$.
\end{lemma}
\begin{proof}
Let $A,B \in \FF$. Then there is an edge in $\de_I(A)$ and an edge in $\de_I(B)$. 
Fig.~\ref{f:gk} depicts all possible locations of (at least one of) such edges. 
In cases (a,b,c) we have $A \cap B, A \cup B \in \FF$ and $A,B$ cross -- 
hence $f$ satisfies the supermodular inequality.
In cases (d,e,f) we have $A \sem B, B \sem A \in \FF$ and $A,B$ co-cross -- 
hence $f$ satisfies the co-supermodular inequality.
\end{proof}

\begin{figure}
\centering 
\includegraphics{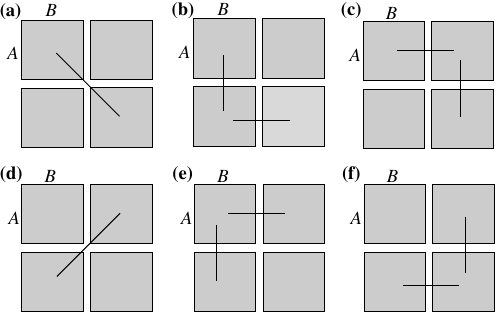}
\caption{Illustration to the proof of Lemma~\ref{l:f}.}
\label{f:gk}
\end{figure}

\noindent
{\bf Remark:} One can show that if $f$ is symmetric crossing supermodular and $\FF$ is {\bf uncrossable} 
($A \cap B,A \cup B \in \FF$ or $A \sem B,B \sem A \in \FF$ whenever $A,B \in \FF$) and $\empt,V \notin \FF$, 
then $f$ is weakly $\FF$-supermodular. 
The proof essentially shows that  $\FF=\{A \subs V:d_I (A) \geq 1\}$ is uncrossable.
Alternatively, we may (roughly) say that if $f$ satisfies {\em both} the supermodular and the co-supermodular inequality 
(whenever appropriate sets are not $\empt,V$), then $f$ is weakly $\FF$-supermodular. 
% the restriction of $f$ to an uncrossable family $\FF$ is weakly supermodular. 

\medskip

\subsection{\emph{k}-EFTS 2-Approximation}
Now we apply the weakly $\FF$-supermodular concept to the \emph{k}-EFTS problem. Note that \emph{k}-EFTS is a particular case of the following meta-problem.

\begin{center} \fbox{\begin{minipage}{0.965\textwidth} \noindent
\underline{\SFC} \\
{\em Input:}  \  \ A graph $G=(V,E)$ with edge costs, set function $f$ defined on a family $\FF \subs 2^V$. \\
{\em Output:}   A min-cost subgraph $H$ of $G$ such that $d_H(A) \geq f(A)$.  
\end{minipage}}\end{center} 

Let $\FF = 2^V \sem\{\empt,V\}$. It is easy to see via Lemma~\ref{l:cover} that for {\sc $k$-EFTS} in particular, we must cover the set function $f(A) = \min\{k,d_G(A)\}$ for $A \in \FF$. 
If we already picked a set $F$ of edges into our solution $H$, then we need to cover the residual function $\min\{d_G(A),k\}-d_F(A)$.
We pick all the {\bf forced edges} -- edges that belong to cuts of $G$ with at most $k$ edges. 
During the iterative rounding algorithm we might pick additional edges. 
So let $F \subs E$ be a set of edges that includes all forced edges. 
Cuts $\de_E(A)$ with $\de_E(A)  \subs F$ cannot and need not be covered, 
hence the function we need to cover is the restriction 
of $f(A)=k-d_F(A)$ (for $A \in 2^V \sem \{\empt,V\}$) to the family $\FF=\{A \subs V:d_{E \sem F} (A) \geq 1\}$.
In the setting of Lemma~\ref{l:f}, we have here $I=E \sem F$. 
The function $k-d_F(A)$ is symmetric crossing supermodular, since $d_F(A)$ is symmetric and submodular. 
Hence by Lemma~\ref{l:f} the function we need to cover is weakly $\FF$-supermodular. 

\begin{theorem}
The $k$-EFTS problem admits approximation ratio $2$ for arbitrary costs 
and $1+4/k$ for unit costs.
\end{theorem}
\begin{proof}
The set function we need to cover is the restriction of the function $f(A)=k-d_F(A)$
to the set family $\FF=\{A \subs V:d_{E \sem F} (A) \geq 1\}$.
By Lemma~\ref{l:f}, we know that $f$ is weakly $\FF$-supermodular. 
The rest of the proof is as in~\cite{ap22}: one can solve LPs in polynomial time and thus apply iterative rounding. 
In the case of arbitrary costs we can use Jain's technique~\cite{Jain01}, while in the case of unit costs Gabow-Gallagher~\cite{GG12}. 
\end{proof}

\end{document}